%% file: main.tex
\documentclass[journal]{IEEEtran}
\usepackage{amsmath,amsfonts}
\usepackage{algorithmic}
\usepackage{algorithm}
\usepackage{array}
\usepackage[caption=false,font=normalsize,labelfont=sf,textfont=sf]{subfig}
\usepackage{textcomp}
\usepackage{url}
\usepackage{verbatim}
\usepackage{graphicx}
\usepackage{amsthm}
\usepackage{cite,graphicx,amsmath,amssymb}
\usepackage{fancyhdr}
\usepackage{hhline}
\usepackage{graphicx,graphics,color,psfrag}
\usepackage{array,color}
\usepackage{mathtools}
\usepackage{amsmath}
\usepackage[T1]{fontenc}
\usepackage{stfloats}
\usepackage{float}  

\usepackage{cite}

\include{header}

\begin{document}

\title{Sparse Array Enabled Near-Field \\ Communications: Beam Pattern Analysis and  Hybrid Beamforming Design}
\author{Cong Zhou, Changsheng You,~\IEEEmembership{Member,~IEEE}, Haodong Zhang, Li Chen,~\IEEEmembership{Senior Member,~IEEE}, Shuo Shi
\thanks{Cong Zhou is with the School of Electronic and Information Engineering, Harbin Institute of Technology, Harbin, 150001, China, and also with the Department of Electrical and Electronic Engineering, Southern Universityof Science and Technology, Shenzhen 518055, China  (e-mail:  zhoucong@stu.hit.edu.cn).}
\thanks{Changsheng You and Haodong Zhang are with the Department of Electronic and Electrical Engineering, Southern University of Science and Technology, Shenzhen 518055, China. (e-mails: youcs@sustech.edu.cn; 12113010@mail.sustech.edu.cn).}
\thanks{Li Chen is with Department of Electronic Engineering and Information Science, University of Science and Technology of China. (e-mail: chenli87@ustc.edu.cn).}
\thanks{Shuo Shi is with the School of Electronic and Information Engineering, Harbin Institute of Technology, Harbin, 150001, China. (e-mail: crcss@hit.edu.cn).}
\thanks{\emph{Corresponding author: Changsheng You.}}
}



\maketitle
\begin{abstract}
Extremely large-scale array (XL-array) has emerged as a promising technology to enable {\it near-field} communications for achieving enhanced spectrum efficiency and spatial resolution, by drastically increasing the number of antennas. However, this also inevitably incurs higher hardware and energy cost, which may not be affordable in future wireless systems. To address this issue, we propose in this paper to exploit two types of {\it sparse arrays} (SAs) for enabling near-field communications. Specifically, we first consider the {\it linear sparse array} (LSA) and characterize its near-field beam pattern. It is shown that despite the achieved beam-focusing gain, the LSA introduces several undesired {\it grating-lobes}, which have comparable beam power with the main-lobe and are focused on specific regions. An efficient hybrid beamforming design is then proposed for the LSA to deal with the potential strong inter-user interference (IUI). Next, we consider another form of SA, called {\it extended coprime array} (ECA), which is composed of two  LSA subarrays with different (coprime) inter-antenna spacing. By characterizing the ECA near-field beam pattern, we show that compared with the LSA with the  same array sparsity, the ECA can greatly suppress the beam power of near-field grating-lobes thanks to the offset effect of the two subarrays, albeit with a larger number of grating-lobes. This thus motivates us to propose a customized two-phase hybrid beamforming design for the ECA. Finally, numerical results are presented to demonstrate the rate performance gain of the proposed two SAs over the conventional uniform linear array (ULA).
\end{abstract}


\vspace{-0.3cm}
\section{Introduction}

The future six-generation (6G) wireless systems are envisioned to not only achieve more stringent performance requirements than the fifth-generation (5G), but also
accommodate a set of new services such as wireless sensing and artificial intelligence (AI) \cite{zhang20196g}. Among others, by drastically increasing the number of antennas at the base stations (BSs), \emph{extremely large-scale arrays} have emerged as a promising technology to significantly boost the spectrum efficiency and spatial resolution in future wireless systems \cite{liu2023near,lu2023tutorial,cui2022near}. In addition, XL-arrays introduce a fundamental change in the wireless channel modelling, shifting from the conventional far-field radio propagation to the new \emph{near-field} counterpart, which brings both new opportunities and challenges\cite{you2023near11}.

Different from the existing works that mostly considered dense XL-arrays with half-wavelength inter-antenna spacing, we study in this paper the new \emph{sparse array} (SA) enabled near-field communication system and characterize its  near-field beam pattern for designing efficient hybrid beamforming.

\vspace{-0.3cm}
\subsection{Prior Works}
\subsubsection{Dense arrays} 
In the existing literature on near-field wireless communications, the dense XL-arrays assumption with half-wavelength antenna spacing has been widely made to facilitate the near-field communication designs and performance analysis. Specifically, for the dense array (DA), the well-known \emph{Rayleigh distance}, which is the boundary between the near- and far-fields, is given by  $Z=(2A^2)/\lambda$, where $ A $ is the array aperture and $ \lambda $ is the carrier wavelength \cite{lu2023tutorial, elbir2023spherical}. Therefore, with a massive number of antennas, the array aperture will be  significantly enlarged, thus rendering the users more likely to be located in the near-field region. In this case, instead of assuming the planar wavefronts in far-field communications, the more accurate \emph{spherical} wavefront propagation model needs to be considered, where the near-field channel steering vector is determined by both the array configuration  and the user location (including both the user angle and \emph{range}). Interestingly, it has been shown that with the spherical wavefronts, near-field beamforming can achievable an appealing function of \emph{beam-focusing}, which focuses  beam energy on a target location/region \cite{zhang2022beam,an2023toward,bjornson2021primer}, hence potentially enhancing the performance of various wireless systems in communication, sensing, power transfer \cite{zhang2023swipt}, etc. Moreover, the \emph{beam-depth} is a crucial metric for assessing the focusing gain in the range domain \cite{kosasih2023finite}. For example, it was revealed in \cite{kosasih2023finite,ref1,5595728} that for multi-user XL-array communication systems, the maximum ratio transmission (MRT) based transmit beamforming can achieve the near-optimal performance when the number of antennas is sufficiently large thanks for the beam-focusing gain. Apart from this, 
%
%
other performance metrics in near-field wireless systems have also been recently studied, such as the signal-to-noise ratio (SNR) scaling order with the antenna number, the degree-of-freedom (DoF) \cite{ji2023extra} in near-field multiple-input multiple-output (MIMO) channels, and so on \cite{liu2023near,lu2023tutorial}.
On the other hand, near-field channel state information (CSI) acquisition is also indispensable for the beamforming design, which, however, is highly challenging due to the new spherical wavefront channel model as well as the larger number of channel parameters to be estimated. To address this issue, the authors in \cite{cui2022channel} proposed a new \emph{polar-domain} channel representation method and then devised efficient near-field channel estimation schemes by using the compressed-sensing techniques. {This framework was further extended in \cite{zhi2023performance,yuan2022spatial,han2020channel} and \cite{chen2023non} to account for the unique spatial non-stationary near-field channel characteristics, where different portions of the XL-array may undergo distinct propagation environment and thus are visible or invisible to the users.}
 Alternatively, near-field beam training is another efficient method to design the XL-array beamforming  without explicit CSI, while its key challenges lie in how to design efficient beam codebooks and devise fast beam training methods. These issues have been extensively   studied recently, such as the two-phase (angle-and-then-range) beam training based on the polar-domain codebook \cite{zhang2022fast}, hierarchical near-field codebook design and beam training \cite{lu2023hierarchical,wu2023two}, joint user angle and range estimation based on the discrete
Fourier transform  (DFT) codebook \cite{wu2023near}, deep-learning (DL) based beam training \cite{liu2022deep}, etc.
\subsubsection{Sparse arrays} While the above works based on DAs can greatly improve the near-field wireless system performance, it also incurs much higher hardware and energy cost due to the huge number of antennas\cite{lu2023tutorial,xu2023low,bjornson2015massive}. To address this issue, several  research efforts have been made in the literature. For example, the authors in \cite{zhang2023dynamic,zhang2022beam} developed efficient and customized near-field hybrid beamforming methods for maximizing the rate performance at low hardware cost. To further reduce the hardware cost, another promising approach is by utilizing SAs (such as the linear SAs, coprime arrays, nested arrays) for enabling near-field communications with high spatial resolution, since the SA aperture can be electrically large at high frequency-bands when the antenna spacing is sufficiently large \cite{yang2023enhancing}. Despite the benefits, SAs also introduce new challenges in dealing with the undesired near-field \emph{grating-lobes}, which may incur severe inter-user interference (IUI) if not properly addressed \cite{wang2023can}. 

In the exiting works on SAs, they have mostly considered the communication designs based on the far-field channel model, or solely focused on the SA-enabled radar sensing. Taking the LSA for instance, the authors in \cite{wang2023can} analyzed its far-field  beam pattern and revealed that the adverse effect of LSA grating-lobes can be mitigated if the users are densely distributed, while the analysis in the near-field region is lacking. Moreover, LSA can be used to create virtual MIMO in radar sensing systems, which is able to achieve an $V_1 V_2$ virtual aperture with only  $(V_1+V_2)$ physical antennas. Next, for the CA which comprises two LSAs with coprime numbers of antennas, the authors in \cite{wang2023can} showed that the CA generally endows more DoFs for the angle-of-arrival (AoA) estimation. Moreover, the far-field CA localization method was further extended in \cite{zheng2020symmetric} to enable the  mixed near- and far-Field source localization with more unique and consecutive virtual sensors and larger physical array aperture. Besides, the authors in \cite{zhou2017robust} proposed an adaptive beamforming algorithm for CAs, which leveraged the high spatial resolution of coprime arrays to estimate the user locations and then recover the interference-plus-noise covariance matrix for designing efficient beamforming.
 In summary, the  important  issue of grating-lobes in SA-enabled near-field communications is still largely uncharted in the literature, which thus motivates the current work. 
\vspace{-0.4cm}
\subsection{Contributions, Organization, and Notations}

In this paper, we consider two types of SAs (with uniform or non-uniform antenna spacing) for enabling near-field communications with a small number of antennas only. Specifically, the first one is the LSA, which has identical inter-antenna spacing larger than the half-wavelength. The second one is the extended CA (ECA), which is composed of multiple basic CAs as shown in Fig.~\ref{fig_1}. In particular, each basic CA consists of two small LSAs with respectively $M$ and $N$ antennas, where $M$ and $N$ are coprime integers.\footnote{For convenience, we consider the ECA in this paper, while the obtained results can be readily extended to other CA configurations, such as  coprime arrays with multi-period subarrays (CAMpS)\cite{wang2018unified}.} For each SA configuration, we first characterize its near-field beam pattern in the main-lobe and grating-lobes, and then propose customized hybrid beamforming methods for maximizing the weighted sum-rate. To the authors' best knowledge, this work is the \emph{first} one to study the beam pattern characteristics and hybrid beamforming design for SA-enabled near-field commutation systems. The main contributions are summarized as follows.
\begin{itemize}
	\item First, we consider the LSA-enabled near-field communication system and characterize its near-field beam pattern. It is shown that for the LSA main-lobe, its beam-width becomes narrower with the increasing array sparsity and the beam-depth is inversely proportional to the array aperture and monotonically reduces when the user is closer to the LSA. Moreover, the LSA with $U$ antennas generate $ 2(U-1) $ grating-lobes in the entire angular domain, which are focused on different angles and ranges. In particular, all the grating-lobes have the same beam-width and beam-height with the main-lobe, while their beam-depths are generally different depending on the user's angle-of-departure (AoD) $ \theta_{0} $. For example, when $ \theta_{0} = 0 $, all the LSA grating-lobes have smaller beam-depths than the main-lobe.
	  Therefore, the LSA is generally susceptible to server IUI, especially when some users are located in the grating-lobes of other users. To maximize the weighted sum-rate, we propose an efficient hybrid beamforming design for LSA that aims to minimizing the difference between the designed hybrid beamforming and the optimal digital beamforming. 
	\item Second, we study the beam pattern characteristic for the ECA-enabled near-field communication system. We show that with approximately the same array sparsity and antenna number, the ECA has the same main-lobe beam-width and beam-depth with the LSA counterpart. Moreover, compared with LSA, the ECA beam pattern generally has more grating-lobes due to the larger inter-antenna spacing of two subarrays than that of the LSA under the same array sparsity. However, the beam-heights of ECA grating-lobes are much smaller than those of LSA grating-lobes, thanks to the offset effect of the grating-lobes generated by the two subarrays, which is thus termed as the \emph{ECA grating-lobes smoothing}. Therefore, the ECA communication systems generally have better worst-case user rate performance than the LSA counterpart, thanks to the reduced interference power at the grating-lobes. Motivated by the above, we then propose a customized hybrid beamforming design for the ECA by first utilizing the MRT beamformer to maximize the received power at each individual user and then designing the digital beamforming to cancel the residual IUI. 
	\item Finally, numerical results are presented to demonstrate the performance gains of the proposed SAs over the conventional DAs (i.e., ULA) as well as the effectiveness of the proposed hybrid beamforming designs. It is shown that given the same (and small) number of antennas, both the LSA and ECA  achieve significant rate performance gains over the ULA system when the users are located at the same or similar angles. This is because the formers have much  larger array apertures and thus enable near-field beam-focusing for reducing the IUI. In addition, when users located at grating-lobes, we show that the ECA achieves a much higher sum-rate than the LSA system due to lower grating-lobe beam-heights.
\end{itemize}

The remainder of this paper is organized as follows. Sections II presents the system model for both the LSA and ECA-enabled near-field communication systems. Then, in Sections III and IV, we characterize the near-field beam pattern for the LSA and ECA, respectively, as well as propose efficient hybrid beamforming designs. Numerical results are presented in Section V to evaluate the performance of the considered SAs, followed by the conclusions drawn in Section VI. 

\textit{Notations}: Lower-case and upper-case boldface letters are employed to represent vectors and matrices, respectively. For vectors and matrices, $ (\cdot)^{T} $ and $ (\cdot)^{H} $ denote the transpose and conjugate transpose operations, respectively. In addtion, $ \left|\cdot\right| $ and $ \left\lVert \cdot\right\lVert  $ respectively indicate the absolute value for a numerical entity and $ l_{2} $ norm.  
\vspace{-0.3 cm}
\section{System Model}
\begin{figure}[!t]\centering
	\includegraphics[width=3.7in]{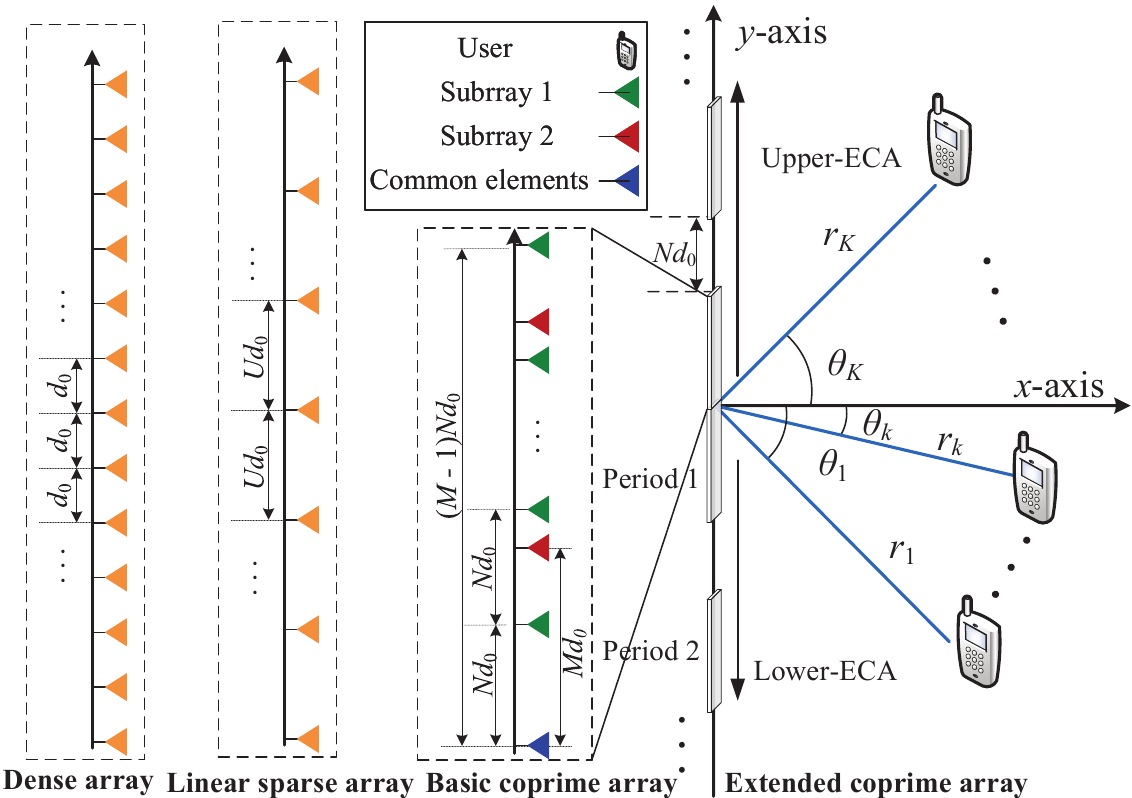}
	\centering
	\caption{Multi-user SA communication systems.}	
	\label{fig_1}
	\vspace{-0.5cm}
\end{figure}

We consider a multi-user SA communication system as shown in Fig.~\ref{fig_1}, where a multi-antenna SA serves $K$ single-antenna users in the downlink. In the following, we first introduce the channel models for the LSA and ECA, respectively, and then present their unified signal model.
\vspace{-0.4cm}
\subsection{Channel Models}

For an SA, we define its array sparsity,  denoted by $\eta_{\rm SA}$, as the ratio between the SA aperture $A_{\rm SA}$ and the ULA aperture $A_{\rm ULA}$ given the same number of antennas. Mathematically, we have  $\eta_{\rm SA} = A_{\rm SA}/A_{\rm ULA}.$
\subsubsection{Linear Sparse Array} The LSA configuration and channel are modeled as follows.

\underline{\textbf{Array model:}} The configuration of an LSA can be characterized by two key parameters, namely, the number of antennas $Q_{\rm LSA}$ (assumed to be an odd number without loss of generality) and the antenna separation parameter $U$. We assume that the LSA is located at the $x$-$y$ axis and centered at the origin. As such, the LSA antennas can be indexed as $q\in \mathcal{Q}_{\rm LSA}\triangleq\{-(\frac{Q \rm_{LSA}-1}{2}),\cdots,0, \cdots,\frac{Q \rm_{LSA}-1}{2}\}$. In addition,  the inter-antenna spacing of LSA is given by $d_{\rm LSA}= U d_0$, where $ U>1$  and $d_0 = \lambda/2$ is the half-wavelength with $\lambda$ denoting the carrier wavelength. Given $\{Q_{\rm LSA}, d_{\rm LSA}\}$, the LSA has an aperture of $A_{{\rm {LSA}}} \!=\!(Q_{\rm LSA}\!-\!1) d_{\rm LSA}=(Q_{\rm LSA}-\!1) U d_0$, for which the corresponding Rayleigh distance is $Z_{\rm LSA}=(2A_{\rm LSA}^2)/\lambda$. For example, given $Q_{\rm LSA} = 101$,  $ U = 4$, and carrier frequency $f = 30$ GHz, we have $Z_{\rm LSA}{=800}$ m, which is much lager than that of a ULA given the same antenna number whose Rayleigh distance is $50$ m only. Therefore, compared with ULA,  users in the LSA communication systems are more likely to be located in the near-field region. 


\underline{\textbf{Channel model:}} Assume that all the users are located in the Fresnel near-field region of the LSA, for which the line-of-sight (LoS) channel follows the uniform spherical wave (USW) model \cite{lu2023tutorial}. Let $r_k$ denote the distance between the LSA center and  user $k\in \mathcal{K}\triangleq\{1,2, \cdots,K\}$. For the USW channel model, we have $1.2A_{\rm LSA}\le r_k\le Z_{\rm LSA}$, where $ 1.2A_{\rm LSA} $ and  $ Z_{\rm LSA} $ are the uniform power distance (UPD) and Rayleigh distance, respectively \cite{lu2023tutorial}. 
We consider the general multi-path channel model, for which the LSA$\to$user $k$ channel, denoted by $\mathbf{h}^H_{{\rm LSA}, k} \in  \mathbb{C}^{1\times Q_{\rm LSA}}$, can be modeled as 
\vspace{-10pt}
\begin{equation}\label{Eq:nf-channel}
\vspace{-10pt}
\mathbf{h}^H_{{{\rm LSA}}, k} = (\mathbf{h}_{{\rm LSA}, k}^{\rm (L)})^H + \sum_{g=1}^{G} (\mathbf{h}_{{\rm LSA}, k, g}^{(\rm NL)})^H,
\end{equation}
which includes one LoS path $ \mathbf{h}_{{\rm LSA}, k}^{\rm (L)} $ and $G$ non-LoS (NLoS) paths $\{\mathbf{h}_{{\rm LSA}, k, g}^{(\rm NL)}\}$. In this paper, we focus on the near-field communications in high-frequency bands such as millimeter-wave (mmWave) and even terahertz (THz), for which the NLoS  paths have negligible power due to severe path-loss and shadowing \cite{wu2023near}. Therefore, the BS-user channel can be approximated as $\mathbf{h}^H_{{\rm LSA}, k} \!\approx\! (\mathbf{h}_{{\rm LSA}, k}^{\rm (L)})^H$.\footnote {The extension to the more general case where the multiple paths have comparable power is more complicated and thus left for future work.}

Based on the USW model, the LSA$\to$user $k$ LoS channel can be modeled as \cite{cui2022channel}
\begin{equation}\label{Eq:nf-model}
	(\mathbf{h}_{{\rm LSA}, k}^{\rm (L)})^H=\sqrt{Q_{\rm LSA}}\beta_{k} \mathbf{b}^H_{\rm LSA}(r_{k}, \theta_{k}),
\end{equation}
where $\beta_k =\frac{\sqrt{\beta_{0}}}{r_{k}}$ is the complex-valued channel gain with $\beta_0$ denoting the reference channel gain at a distance of $r_{0} = 1$ m; $ \theta_{k}\in [-\frac{\pi}{2},\frac{\pi}{2}] $ is AoD from the LSA center to user $k$.
Moreover, $\mathbf{b}^H_{\rm LSA}(r_{k}, \theta_{k})\in \mathbb{C}^{1\times Q_{\rm LSA} }$ denotes the near-field  LSA channel steering vector  with
\begin{equation}\label{near_steering}
\left[\mathbf{b}^H_{\rm LSA}\left(r_{k},\theta_{k}\right) \right]_{q} = \frac{1}{\sqrt{Q_{\rm LSA}}}e^{-\frac{\jmath 2 \pi r_{k,q}}{\lambda}}, \forall q\in \mathcal{Q}_{\rm LSA},
\end{equation}
where $r_{k, q}$  denotes the distance between user $k$ and  antenna $q$. Mathematically, based on the Fresnel approximation, $r_{k, q}$ can be approximated as   \cite{cui2022channel}
\begin{align}
r_{k, q} &= \sqrt{r_{k}^{2} -2 q  U d_0 r_{k} \sin \theta_k + (q U d_0)^2}\nn \\
&\approx r_{k}-q U d_0 \sin \theta_{k}+\frac{q^2 (  U d_0) ^2 \cos^2 \theta_{k}}{2 r_{k}}.
\end{align}


\subsubsection{Extended Coprime  Array} The ECA configuration and its channel model are given as follows.

\underline{\textbf{Array model:}} As shown in Fig.~\ref{fig_1}, the ECA consists of $L$ (assumed to be an even number) \emph{basic} CAs.
Each basic CA is composed of two small LSAs, including an $M$-antenna LSA with $N d_0$ inter-antenna spacing and an $N$-antenna LSA with $M d_0$ inter-antenna spacing, where $M$ and $N$ are coprime integers with $M\ge N$. Note that the two small LSAs share the first antenna (see Fig.~\ref{fig_1}); thus each basic CA has $(M+N-1)$ antennas in total with an aperture of  $(M-1)Nd_0$ \cite{8472789Zhou}. 

The entire ECA can be divided into two \emph{half}-ECAs which are symmetric with respect to (w.r.t) the origin. Specifically, the upper-ECA includes $L/2$ basic CAs with an inter-CA spacing of $Nd_{0}$ to enforce the periodic extension \cite{wang2018unified}. As such, the upper-ECA has $(L/2)(M+N\!-\!1)$ antennas in total with an aperture of $A_{\rm uECA}=\frac{L}{2}(M-1)Nd_0+ (\frac{L}{2}-1) Nd_{0}$.
The lower-ECA can be similarly modeled via array configuration symmetry, while it shares a common antenna  at the origin with the upper-ECA. 
Hence, the entire ECA has a total number of $Q_{\rm ECA}=L(M+N-1)\!-\!1$ antennas with the array sparsity of $\eta_{\rm ECA} = \frac{MN}{M+N-1}$ and the aperture size of
\begin{equation}
A_{\rm ECA}=2A_{\rm uECA} = (LM-2)Nd_0.
\end{equation}
For an ECA system setup with $L=12 $,  $M = 7$, $ N=5$ and $f = 30$ GHz, the ECA has $131$ antennas and its Rayleigh distance is  $Z_{\rm ECA}=\frac{2A_{\rm ECA}^2}{\lambda}\approx 834$ m. 
Moreover, it is worth noting that ECA can be regarded as a general array configuration that includes the ULA and LSA as two special cases for which we have  $\{ M = 1 ~{\rm or}~ N = 1\}$ and $M=N$, respectively. 

%
%


\underline{\textbf{Channel model:}} 
Similar to LSA, we assume that all users are located in the Fresnel near-field region of the ECA. Under the USW channel model, the near-field channel from the ECA to user $k$ can be modeled as
\begin{equation}\label{Eq:nf-model}
		\mathbf{h}^H_{{\rm ECA},k}\approx (\mathbf{h}_{{\rm ECA}, k}^{(\rm L)})^H=\sqrt{Q_{\rm ECA}}\beta_{k} \mathbf{b}^H_{\rm ECA}(r_{k}, \theta_{k}).
\end{equation}
Therein, $\mathbf{b}^H_{\rm ECA}\left(r_{k},\theta_{k}\right)\in \mathbb{C}^{1\times {Q_{\rm ECA}}} $ denotes the ECA channel steering vector, which is jointly determined by the user location $\{r_{k},\theta_{k}\}$ and the ECA configuration $\{M, N, L\}$. 

To facilitate the modelling for $\mathbf{b}^H_{\rm ECA}\left(r_{k},\theta_{k}\right)$, we equivalently re-divide the entire ECA into two sparse \emph{subarrays} as illustrated in Fig.~\ref{fig_1}, for which subarray $1$ includes $(LM\!-\!1)$ antennas with $N d_0$ inter-antenna spacing and subarray $2$ includes $(LN\!-\!1)$ antennas with $Md_0$ inter-antenna spacing (both can be regarded as an \emph{effective large} LSA); while these two subarrays share $ (L-1)$ antennas.  
Mathematically, the $m$-th antenna of subarray $1$ is located at ($0, mNd_0$), $m\in \mathcal{M}\triangleq\{-\frac{LM}{2}+1,\cdots,0, \cdots,\frac{LM}{2}-1\}$, the  $n$-th  antenna of subarray $2$ is located at ($0, nMd_0$), $n\in \mathcal{N}\triangleq\{-\frac{LN}{2}\!+\!1,\cdots,0, \cdots,\frac{LN}{2}\!-\!1\}$, and the two subarrays have common antennas at ($ 0, iMNd_0 $), $i\in \mathcal{I}\triangleq\{-\frac{L}{2}+1,\cdots,\!0, \!\cdots\!,\frac{L}{2}\!-\!1 \}$. 

As such, by denoting
\vspace{-3pt} 
$$\bar{\mathbf{b}}^H_{\rm ECA}\left(r_{k},\theta_{k}\right)\!=\!\left[\mathbf{b}^H_{1}\left(r_{k},\theta_{k}\right), \mathbf{b}^H_{2}\left(r_{k}, \theta_{k}\right) \right]$$
 as the \emph{effective} ECA channel steering vector composed of the steering vectors of two subarrays, $\mathbf{b}^H_{\rm ECA}\left(r_{k},\theta_{k}\right)$ can be equivalently expressed as\vspace{-0.1cm} 
\begin{equation}
\mathbf{b}^H_{\rm ECA}\left(r_{k},\theta_{k}\right)= \bar{\mathbf{b}}^H_{\rm ECA}\left(r_{k},\theta_{k}\right)\times \Pi,
\end{equation}
where $\Pi $ denotes the array permutation matrix determined by the ECA configuration.
For subarray $1$, each entry of $\mathbf{b}^H_{1}(r_{k},\theta_{k})$ is given by
\vspace{-5pt}
\begin{equation}
\left[\mathbf{b}^H_{1}\left(r_{k},\theta_{k}\right) \right]_{m} = \dfrac{1}{\sqrt{Q_{\rm ECA}}}e^{-\frac{\jmath 2 \pi r_{k, m}^{(1)}}{\lambda}}, \forall m\in \mathcal{M},
\end{equation}
where $r_{k, m}^{(1)} \approx r_{k}-mNd_0 \sin \theta_{k}+\frac{m^2 ( Nd_0) ^2 \cos^2 \theta_{k}}{2 r}$ is the distance between user $k$ and the $m$-th antenna of subarray $1$. Similarly, $\mathbf{b}^{H}_{2}\left(r_{k}, \theta_{k}\right)$ can be modeled.

\vspace{-0.3cm} 
\subsection{Signal Model}

For both the above two SA downlink communication systems, we consider the low-cost hybrid beamforming architecture with  $K$ RF chains.  Then the received signal at user $k$ is given by
\begin{equation}
	y_k=\mathbf{h}_k^{H} \mathbf{F}_{\mathrm{A}} \mathbf{F}_{\mathrm{D}} \mathbf{x}+z_k,
\end{equation}
where $  \mathbf{x} \in \mathbb{C}^{K \times 1} $ denotes the transmitted symbols for $ K $ users, $ z_{k} \sim  \mathcal{C N}\left(0, \sigma_k^2\right)$ represents the additive white Gaussian noise (AWGN) with power $ \sigma_k^{2} $, and $ \mathbf{F}_{\mathrm{A}} \in \mathbb{C}^{Q \times K} $ and $ \mathbf{F}_{\mathrm{D}} \in \mathbb{C}^{K \times K} $ are the analog and digital beamformers, respectively. In addition, $\mathbf{h}_k^{H}$ is the unified channel from the SA to user $k$ for the two SAs,  which reduces to $\mathbf{h}^H_{{\rm LSA},k}$ for the LSA  and $\mathbf{h}^H_{{\rm ECA},k}$ for the ECA, respectively.  
Based on the above, the received signal-to-interference-plus-noise ratio (SINR) at user $k$ is given by
\begin{equation}\label{eq11}
	 {\rm SINR}_k=\frac{\left|\mathbf{h}_k^{H} \mathbf{F}_{\mathrm{A}} \mathbf{f}_{\mathrm{D}, k}\right|^2}{\sum\limits_{i = 1, i \neq k}^{K}\left|\mathbf{h}_k^{H} \mathbf{F}_{\mathrm{A}} \mathbf{f}_{\mathrm{D}, i}\right|^2+\sigma_k^2}, ~~\forall k\in \mathcal{K},
\end{equation}
where $ \mathbf{f}_{\mathrm{D}, k} $ represents the $k$-th column of $ \mathbf{F}_{\mathrm{D}} $.

\vspace{-0.4cm}
\subsection{Main Definitions}
To facilitate the subsequent SA beam pattern analysis, we make the following main definitions.
Let ${\bf{b}}(r_0, \theta _0) $ denote a general near-field beamforming steering vector for a target location $\{{r_0},{\theta _0}\}$, and ${\bf{b}}({r},{\theta })$ denote the channel steering vector at an observation location $\{r, \theta \}$. The general near-field beam pattern, beam-width,  beam-depth, and beam-height for SAs are defined as follows. Note that the null-to-null beam-width is adopted for obtaining closed-form expressions to gain useful insights, while it can be also extended to obtain the 3-dB beam-width yet without closed forms in general. These definitions have been widely used in the literature \cite{lu2023tutorial}.
\begin{definition}[Beam pattern]\label{Def:BeamPattern}
	\rm {The beam pattern of ${\bf{b}}(r_0, \theta _0)$  characterizes its (normalized) beam power at an arbitrary observation  location $\{r, \theta\}$ in the considered near-field region, which is defined as}
	\begin{align}
	f({r_0},{\theta _0};r,\theta )& \triangleq {\left| {{{\bf{b}}_{\rm }^H}({r_0},{\theta _0}){\bf{b}}(r,\theta )} \right|}, \nn \\
	\quad \quad &\forall r\in  (1.2 A_{\rm SA},Z_{\rm SA}) , \theta \in \l[-\pi/2, \pi/2\r],
	\end{align}
	where $A_{\rm SA}$ and $Z_{\rm SA}$ are the SA aperture and Rayleigh distance for the SAs, respectively.
\end{definition}
\vspace{-0.3cm} 
\begin{definition}[Beam-width]\label{Def:width}
	\rm The null-to-null beam-width is the spatial angular width in the \emph{user-ring} (which is defined as $ \frac{\cos^2 \theta}{r} = \frac{\cos^2 \theta_0}{r_0} $ and will be explained later in more details), from which the magnitude of  beam pattern decreases to zero (respectively denoted as $ \theta_{\rm right} $ and $ \theta_{\rm left} $) away from the main-lobe (or grating-lobe), i.e.,  
	\begin{equation}
		{\rm BW}  \triangleq \left| \sin \theta_{\rm right} - \sin \theta_{\rm left} \right|.
	\end{equation}
\end{definition}
\vspace{-0.4cm} 
\begin{definition}[Beam-depth]\!\!\label{Def:Depth}
	\rm For an observation angle $\theta$, the 3-dB beam-depth characterizes the length of range interval $r\in [r_{\rm large}, r_{\rm small}]$, for which it satisfies $ f({r_0},{\theta _0}; r, \theta)\le  \frac{1}{2} \max_{r'\in  (1.2 A_{\rm SA},Z_{\rm SA})}\{f({r_0},{\theta _0};r',\theta )\}$. Thus, the beam-depth is given by
	\vspace{-0.2cm} 
	\begin{equation}
		{\rm BD}  \triangleq \left| r_{\rm large} - r_{\rm small} \right|.
	\end{equation}
\end{definition}
\vspace{-0.2cm} 
\begin{definition}[Beam-height]\label{Def:height} 
	{\rm The beam-height characterizes the magnitude  of the main-lobe (or grating-lobes) at an observation angle $\theta$, defined as}
	\vspace{-0.2cm} 
	\begin{equation}
		{\rm BH}  \triangleq \max_{r\in  (1.2 A_{\rm SA},Z_{\rm SA})}\{f({r_0},{\theta _0};r,\theta )\}.
	\end{equation}
\end{definition}
\vspace{-0.2 cm}
Note that different from conventional ULAs, the beam pattern of SAs usually includes both the main-lobe and additional \emph{grating}-lobes; thus we define the beam-width and beam-depth accounting for both the main- and grating-lobes. 
\vspace{-0.1cm} 
\section{Linear Sparse Array}
In this section, we characterize the near-field LSA  beam pattern and propose an efficient hybrid beamforming design.

First, the general near-field LSA beam pattern is obtained below, which can be  proved by using similar methods in \cite{cui2022channel}.
\begin{lemma}\emph{For an LSA parameterized by $\{Q_{\rm LSA}, U\}$, the beam pattern of $\mathbf{b}_{\rm LSA}(r_{0}, \theta_{0})$ is given by 
{\begin{align}
			&f_{\rm LSA}  \left( r_0, \theta_0; r, \theta \right) \nn\\
		\!\!\! \!\!\! 	\overset{a_1}{\approx}& \frac{1}{Q_{\rm LSA}} \left| \sum_{q\in \mathcal{Q}_{\rm LSA}}\!\!\!\!\exp{\l(\jmath  \underbrace{\frac{2\pi}{\lambda} q  U d_{0} \Delta}_{B_1}+\jmath \underbrace{\frac{\pi}{\lambda} q^2 ( U d_{0})^2\Phi }_{B_2 }\r)} \right|,\nn\\
			\triangleq&  \hat{f}_{\rm LSA}  \left( r_0, \theta_0; r, \theta \right)\!\label{LSAsum}
	\end{align}}
where $ \Delta\!\triangleq\! \sin\theta\!-\!\sin\theta_0$ is defined as the \emph{spatial angle difference}, $ \Phi \triangleq\!  \frac{{{{\cos }^2}\theta_{0} }}{{r_{0}}} \!-\! {\frac{{{{\cos }^2}{\theta }}}{{{r}}}}$ is named as the \emph{ring difference}, and $(a_1)$ is due to the Fresnel approximation which has been shown to be accurate in \cite{kosasih2023finite}. 
}
\end{lemma}
\begin{figure}[t]
	\centering
	\includegraphics[width=2.7in]{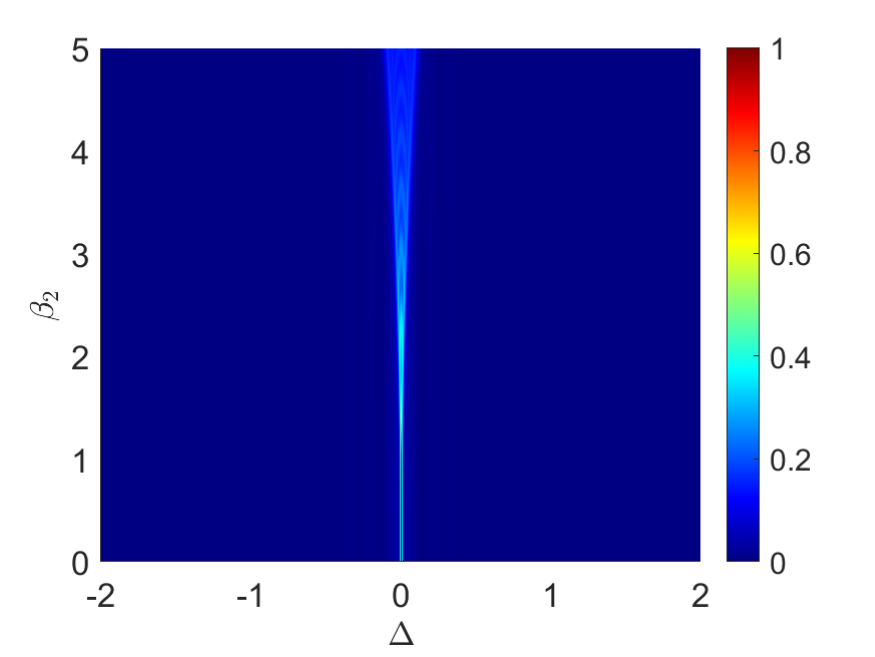}\label{fig:subfig8}
	\caption{Illustration of the function $G(\cdot)$ w.r.t. $ \Delta $ and $ \beta_{2} $.}
	\label{fig3}
	\vspace{-0.5cm}
\end{figure}
\begin{figure*}[!t]
	\centering
	\captionsetup[ subfloat]{labelfont=rm, format=plain,labelformat=empty}	
	\subfloat[{\small {\rm ULA with} $ 119 $ {\rm antennas}.}]{\includegraphics[width=2.33in]{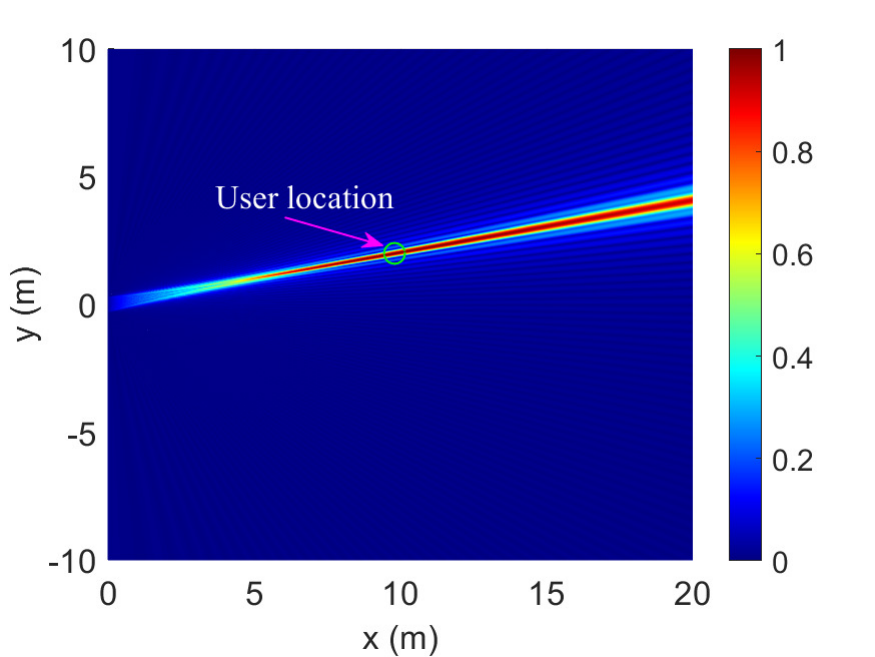} \label{fig_first_case}}
	\hfil
	\subfloat[{\small {\rm LSA with} $ U = 2$ {\rm and }  $Q_{\rm LSA} = 119$}.] {\includegraphics[width=2.33in]{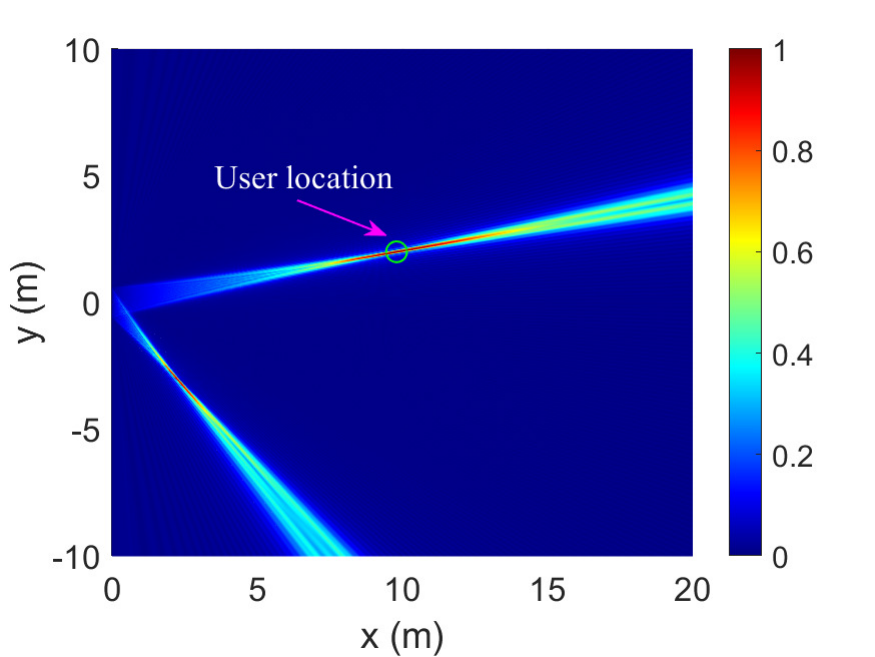} \label{fig_second_case}}
	\hfil
	\subfloat[{\small {\rm ECA with} $M = 3, N = 4 $, $ Q_{\rm ECA} = 119$. }]{\includegraphics[width=2.33in]{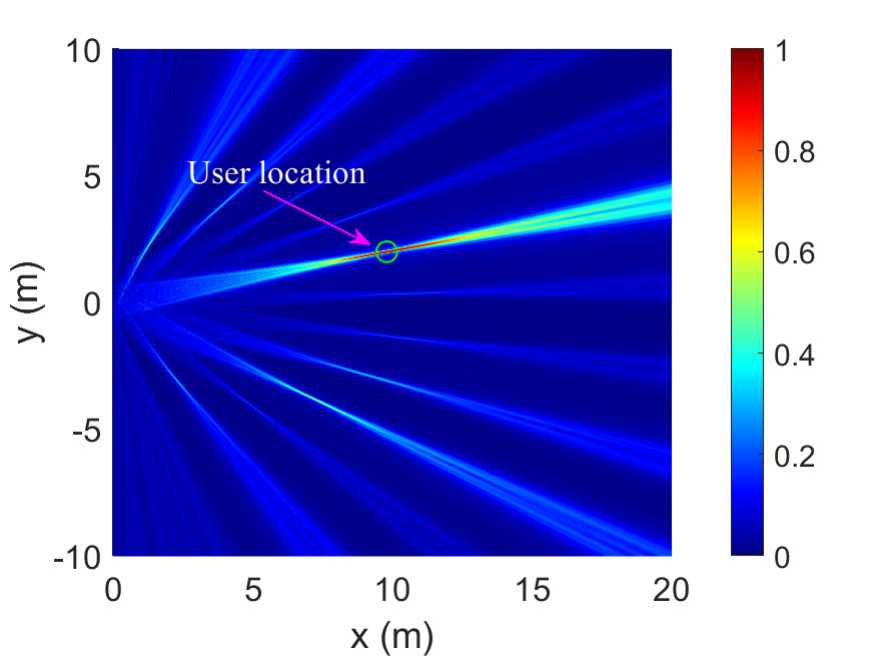} \label{fig_third_case}}
	\caption{{Beam patterns of different array configurations under $ \lambda = 0.01$~m at $(\sin\theta_0 = 0.2, r_0 =  10~{\rm m}) $.}}
	\label{fig_4}
	\vspace{-0.5cm} 
\end{figure*}

Compared with the near-field ULA  beam pattern given in \cite{liu2023near}, that for LSA mainly differs in the inter-antenna spacing ($U d_0$ versus $d_0$), which, however, will be shown to have significant impact on the beam pattern. To gain useful insights, notice that
$B_1$ and $B_2$ defined in \eqref{LSAsum} have the following properties 
\vspace{-0.1cm}
\begin{align}
	\begin{cases}
		B_1=2{\tau } \pi,  \tau\in \mathbb{Z},& \text{when}~ \Delta=\frac{2 u}{ U},  u\in \mathcal{U},\\
		B_2= 0, & \text{when}~\Phi=\frac{{{{\cos }^2}\theta_{0} }}{{r_{0}}}\! -\! {\frac{{{{\cos }^2}{\theta }}}{{{r}}}}=0,
	\end{cases}
\end{align}
\vspace{-0.1cm}where $ \mathcal{U} \triangleq \{\pm 1,\dots,\pm (U-1)\}$. 
Then, we divide the LSA beam pattern analysis into the following three cases:
\begin{itemize}
\item {\bf Case 1: }  $\Phi=0$. For this case, we have
$f_{\rm LSA}=\frac{1}{Q_{\rm LSA}}| \sum_{q\in \mathcal{Q}_{\rm LSA}} \exp{(\jmath  \frac{2\pi}{\lambda} q  U d_{0} \Delta )} |.$
This is the beam pattern at a ring characterized by $\{(r, \theta) | \frac{\cos^2 \theta}{r} =  \frac{\cos^2 \theta_0}{r_0}\}$, which contains the user location and thus is called the \emph{user-ring}. Its beam pattern is similar to the far-field LSA case \cite{wang2023can}, for which the beam-width of the main-lobe and grating-lobes in the ring can be obtained.
\item {\bf Case 2:} $\Delta =\! \frac{2 u}{ U}$, which corresponds to the angles of grading-lobes and this case will be used for obtaining the beam-depths of the grating-lobes.
\item {\bf Case 3:} $\Phi\neq 0$ and  $\Delta\neq \frac{2 u}{ U}$: The expression of  $f_{\rm LSA}$ is still given in \eqref{LSAsum}. For this case, we will show below that in most cases, the beam pattern $f_{\rm LSA} \left( r_0, \theta_0; r, \theta \right)\approx 0$.
\end{itemize}


%
%
%
%
%

\begin{lemma}\label{LSAgeneral}\emph{
	If $\Phi\neq 0$ and  $\Delta\neq \frac{2 u}{ U}, \forall u\in \mathcal{U}$, then $ \hat{f}_{\rm LSA}  \left( r_0, \theta_0; r, \theta \right)$ in \eqref{LSAsum} can be approximated as
	\begin{equation}\label{LSAClosedForm}
	\begin{aligned}
	\hat{f}_{\rm LSA}  \left( r_0, \theta_0; r, \theta \right) \approx \left| G(\beta_{1},  \beta_{2})\right|,
	\end{aligned}		
	\end{equation}
	 where $G(\beta_{1},  \beta_{2}) \triangleq (\widehat{C}(\beta_{1},\beta_{2}) +  \jmath(\widehat{S}(\beta_{1},\beta_{2}))/(2\beta_2)$,
	$ \widehat{C}(\beta_{1},\beta_{2}) \triangleq {C}(\beta_{1}+\beta_{2}) - C(\beta_{1}-\beta_{2})$ {\rm and} $ \widehat{S}(\beta_{1},\beta_{2}) \triangleq S(\beta_{1}+\beta_{2}) - S(\beta_{1}-\beta_{2}) $. Specifically, $ C(x) = \int_{0}^{x} \cos(\frac{\pi}{2}t^2 ){\rm d}t $ {\rm and} $ S(x) = \int_{0}^{x} \sin(\frac{\pi}{2}t^2 ){\rm d}t $ {\rm denote the Fresnel integrals. $ \beta_1 $ and $ \beta_2 $ are defined as}
	\begin{equation}\label{beta}
	\beta_{1} = \frac{\Delta}{{\sqrt{d_0\left| \Phi  \right|}}},
	~~
	\beta_{2} = \frac{Q_{\rm LSA} U}{2}\sqrt{{d_0}\left| \Phi \right| }.
	\end{equation}}
\end{lemma}
\begin{proof}
	Please refer to \bf Appendix A.
\end{proof}

Note that from (\ref{beta}), we can obtain $$ \beta_{1}\beta_{2} = \frac{Q_{\rm LSA} U\Delta}{2}=\frac{A_{\rm LSA}}{2d_0}  \Delta.$$ Therefore, the function $G(\cdot)$ can also be expressed as a function of $\{\Delta, \beta_2\}$. In Fig.~\ref{fig3}, we consider the system setup of \{$ Q_{\rm LSA} = 131 $, $ d_{\rm LSA} = 3d_{0} = 0.03$ m\} and numerically show the  function of $G(\cdot)$. Importantly, It is observed that except for the scenario of $\Delta\to0$, in most cases for $\{\Delta, \beta_2\}$, the LSA beam pattern $\hat{f}_{\rm LSA}  \left( r_0, \theta_0; r, \theta \right) \approx 0$. This indicates that for the LSA beam pattern analysis under $\Phi\!\neq\! 0$ and  $\Delta\neq \frac{2 u}{ U}$, we only need to examine the case where $ \Delta\!\to\! 0$, which may yield non-zero beam pattern.

\subsection{Main-lobe of LSA Beam Pattern }
In this subsection, we characterize the beam-width and beam-depth for the main-lobe of LSA beam pattern.

\subsubsection{{Beam-width in the Angular Domain}}
To obtain the main-lobe beam-width, we consider the beam pattern at the user-ring, which essentially characterizes the beam power at different angles given the angle-dependent distances, i.e., $ \frac{\cos^2 \theta}{r} =  \frac{\cos^2 \theta_0}{r_0}$. Then, the LSA beam pattern in Case 1) can be obtained as below.
%
\begin{lemma}\label{lemma3}\emph{
	 When $\Phi=0$ (or equivalently $ \frac{\cos^2 \theta}{r} =  \frac{\cos^2 \theta_0}{r_0}$), $\hat{f}_{\rm LSA}({r_0},{\theta _0};r,\theta) $  in \eqref{LSAsum}  is given by
		\begin{equation}\label{eq27}
			\begin{aligned}
				\hat{f}_{\rm LSA}({r_0},{\theta _0};r,\theta ) = \frac{1}{Q_{\rm LSA}}
				\left| \Xi_{Q_{\rm LSA}}(U \Delta ) \right|,
			\end{aligned}		
	\end{equation}
	 where $ \Xi_{\alpha}(x) \triangleq \frac{\sin(\frac{\alpha x \pi}{2})}{\sin(\frac{x \pi}{2})}$  is the Dirichlet sinc function.}
\end{lemma}
\begin{proof}
	Similar to \cite{wang2023can}, for an LSA, {\rm if} $ \frac{\cos^2 \theta}{r} =  \frac{\cos^2 \theta_0}{r_0} $, its beam pattern can be rewritten as
	\begin{equation}
		\begingroup\makeatletter\def\f@size{9}\check@mathfonts
		{\small
		\begin{aligned}
			{\bf{b}}_{\rm LSA}^{\rm{H}}({r_0},{\theta _0}){{\bf{b}}}_{\rm LSA}(r,\theta )= 
			\frac{{e^{\jmath\frac{{2\pi }}{\lambda }({r_0} - r)}}\Xi_{Q_{\rm LSA}}(\Delta U)}{Q_{\rm LSA}}.
		\end{aligned}}
		\endgroup\nn
	\end{equation}
	
	Therefore, we have the desired result. 
\end{proof}

To characterize the null-to-null beam-width, we set $ {\frac{1}{2}Q_{\rm LSA} U\pi \Delta} = \pm \pi $ and attain the following result.

\begin{proposition}[Beam-width of LSA main-lobe]\emph{For an LSA parameterized by $\{Q_{\rm LSA}, U\}$, the null-to-null beam-width for the  main-lobe of LSA beamformer ${{\bf{b}}_{\rm LSA}}({r_0},{\theta _0})$ is 
\begin{equation}\label{LSA beam-width}
{\rm BW}^{(\rm M)}_{{\rm LSA}}=\frac{4}{Q_{\rm LSA} U}\approx\frac{4d_0}{A_{\rm LSA}}.
\end{equation}
}
\end{proposition}

%
%

It is observed from \eqref{LSA beam-width} that the beam-width of LSA main-lobe is \emph{inversely} proportional to the LSA aperture. As the LSA aperture is usually much larger than that of ULA given the same antenna number,  the beam-width of LSA main-lobe is significantly narrower than that of ULA. Thereby, the LSA can more effectively mitigate the inter-user interference, especially in scenarios with  dense user distribution.
                                                                                                                                                                                                                                                                       
\subsubsection{Beam-depth in the Range Domain}~                                                                                                                                                                                                                        
For the LSA main-lobe, its beam-depth is another important parameter for characterizing the beam-focusing effect in the range domain. Based on Definition~\ref{Def:Depth}, we obtain the  beam-depth of  LSA main-lobe, which can be  proved by using similarly methods in \cite{liu2023near} and thus omitted for brevity.
\begin{lemma}\emph{When $ \theta = \theta_{0} $ (equivalently $ \beta_{1} = 0$), the LSA beam pattern in \eqref{LSAClosedForm} reduces to
\begin{equation}
{\hat{f}_{\rm LSA}({r_0},{\theta _0};r,\theta ) \approx}
	\left| F(\beta_{2})\right|,
\end{equation}
where $F(x)  \triangleq \frac{C(x)+\jmath S(x)}{x}$ with $\beta_{2} = \frac{Q_{\rm LSA} U}{2}\sqrt{{d_0}\left| \Phi \right|}$.
}
\end{lemma}

\begin{proposition}[Beam-depth of LSA main-lobe]\label{main-lobe rang}
	\emph{
For an LSA parameterized by  $\{Q_{\rm LSA}, U\}$,  the 3-dB  beam-depth for the main-lobe of LSA beamformer ${\bf{b}}_{\rm LSA}({r_0},{\theta_0 })$ is given by
\begin{equation}
		\mathrm{BD}_{{\rm LSA}}^{\rm (M)} = \begin{cases}\dfrac{2 r_{0}^2 r_{\mathrm{LSA}}}{r_{\mathrm{LSA}}^2-r_{0}^2}, & r_{0}<r_{\mathrm{LSA}} \\ \infty, & r_{0} \geq r_{\mathrm{LSA}}\end{cases},
	\end{equation}
	\rm where $ r_{\mathrm{LSA}} {\approx} \frac{Q_{\rm LSA}^{2} U^{2}d_0 \cos ^2 \theta_{0}}{4  \varphi _{3 \mathrm{dB}}^2}$ and $ \varphi _{3\rm dB} = 1.6 $\cite{cui2022channel}. 
}
\end{proposition}

\begin{remark}[What affects the beam-depth of LSA main-lobe?]
\emph{Proposition \ref{main-lobe rang} indicates for the LSA communication system, there appears the beam-focusing effect only when the user range is smaller than a threshold, which is determined by the LSA sparsity, antenna number, and user angle. Generally speaking, with a higher array sparsity, there is more prominent beam-focusing effect than the conventional ULA. 
	In addition, as the user's AoD, $ \theta_0 $, increases, the focused region diminishes. In particular, when $ \theta_0 = \pm \frac{\pi}{2} $, the focused region becomes zero.	
%
Moreover, when $ r_{0}<r_{\mathrm{LSA}} $, the beam-depth of LSA main-lobe is inversely proportional to the aperture and monotonically reduces when the user is closer to the LSA.
}
\end{remark}

\vspace{-0.7cm}
\subsection{Grating-lobes of LSA Beam Pattern}
\subsubsection{Beam-width in the Angular Domain}
Although the LSA can enhance the near-field beam focusing effect by greatly enlarging the array aperture, it usually introduces undesired grating-lobes that may aggravate the inter-user interference issue. In this subsection, we characterize the beam-width, beam-depth, and beam-height for the grating-lobes of the LSA. 

First, we determine the angles for the LSA grating-lobes as below, which can  be proved from the condition $ \Delta =\! \frac{2 u}{ U}$.
%
\begin{lemma}\label{LemmaLSAGratingn-lobe}
	\rm For an LSA parameterized by $\{Q_{\rm LSA}, U\}$, the grating-lobes of its beamformer ${{\bf{b}}_{\rm LSA}}({r_0},{\theta _0})$ occur at the following angles:
	\begin{equation}
	 \theta^{(\rm G)}_{{\rm LSA}, u} \!\!=\! \arcsin\l(\sin\theta_0\!+\!\frac{2u}{ U}\r), \forall u = \pm 1, \!\dots,\! \pm ( U - 1).
	 	\end{equation}
		Moreover, the beam-height of all LSA grating-lobes is $1$.
\end{lemma}
\begin{figure}[t]
	\begin{center}
		\centering
		\includegraphics[width=2.7in]{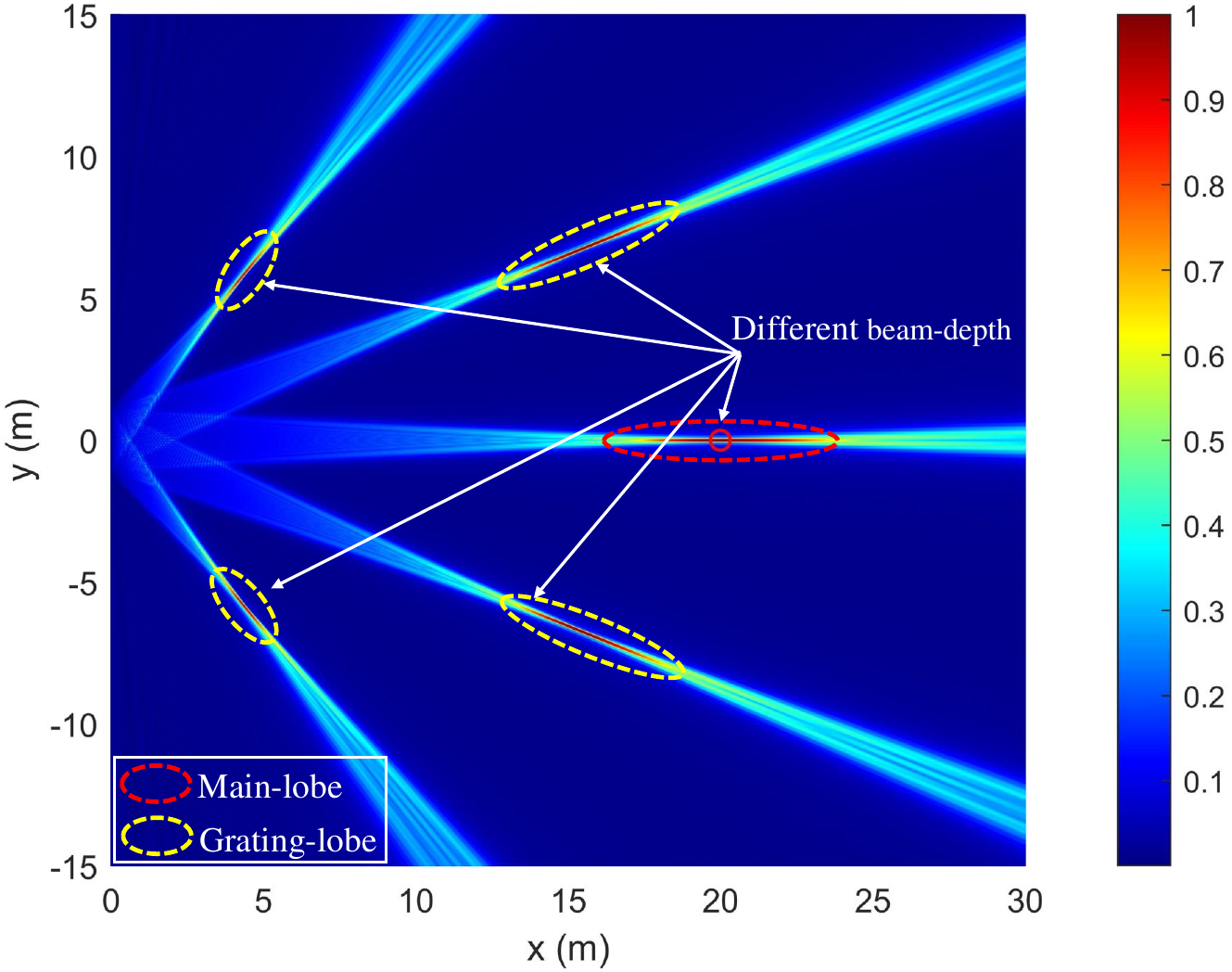}
		\caption{{Beam-depths of main-lobe versus grating-lobes with $ Q_{\rm LSA} = 129 $ and $ U = 5 $ under $ \lambda = 0.01 $ m at ($ 0^{\circ} $, 20 m).}} 
		\label{fig_5}
		\vspace{-0.5cm} 
	\end{center}	
\end{figure}
\vspace{-0.1cm}
Lemma \ref{LemmaLSAGratingn-lobe} shows that there appear {$2(U\!-\!1)$} strong grating-lobes in the LSA beam pattern with a beam-height of $1$ in the entire angular domain as illustrated in Figs. \ref{fig_4}(a). Further, similar to the LSA main-lobe, the beam-width of LSA grating-lobes can  be easily obtained as $ {\rm BW}^{(\rm G)}_{{\rm LSA},u}=\frac{4d_0}{A_{\rm LSA}}, \forall u. $


\subsubsection{Beam-depth in the Range Domain}
For the beam-depth analysis of LSA grating-lobes, we set $\theta = \theta_{{\rm LSA},u}^{\rm (G)}$ and thus obtain the following result according to Definition~\ref{Def:Depth}.
\begin{lemma}\label{Lemmma5}\emph{
	When $ \Delta= \frac{2u}{ U}, \forall u\in \mathcal{U}$, the LSA beam pattern $ \hat{f}_{\rm LSA}({r_0},{\theta _0};r,\theta_{{\rm LSA},u}^{\rm (G)} )$ in \eqref{LSAsum} can be approximated as
	\begin{equation}
	\begin{aligned}
	\hat{f}_{\rm LSA}({r_0},{\theta _0};r,\theta_{{\rm LSA},u}^{\rm (G)} ) \approx 
	\left| F(\beta_{2}) \right|, \forall u\in \mathcal{U},
	\end{aligned}		
	\end{equation}
	\rm where {$ \beta_{2} = \frac{Q_{\rm LSA} U}{2}\sqrt{{d_0}\left| \frac{{{{\cos }^2}\theta_{0} }}{{r_{0}}} - {\frac{{{{\cos }^2}{\theta_{{\rm LSA},u}^{\rm (G)}}}}{{{r}}}} \right| } $}.}
\end{lemma}
\begin{proof}
	The prove is similar in \cite{cui2022channel} and thus are omitted. 
\end{proof}


\begin{proposition}[Beam-depth of LSA grating-lobes]\label{LSA GL Beam-depth}
	\rm For an LSA parameterized by  $\{Q_{\rm LSA}, U\}$,  the 3-dB  beam-depths for the grating-lobes of LSA beamformer ${\bf{b}}_{\rm LSA}({r_0},{\theta_0 })$ are 
	\vspace{-8pt}
	\begin{equation}
		\mathrm{BD}_{{\rm LSA}, u}^{\rm (G)}= \begin{cases}\frac{2 r_{0}^2 r_{\mathrm{LSA},u}^{\rm (G)}}{r_{\mathrm{LSA}}^2-r_{0}^2}, & r_{0}<r_{\mathrm{LSA}} \\ \infty, & r_{0} \geq r_{\mathrm{LSA}}\end{cases}, \forall u,
		\vspace{-6pt}
	\end{equation}
	\rm where $ r_{\mathrm{LSA},u}^{\rm (G)} \approx \frac{
		Q_{\rm LSA}^{2} U^{2}d_0 \cos ^2 \theta^{\rm (G)}_{{\rm LSA},u}}{4  \varphi _{3 \mathrm{dB}}^2}$, $\forall u$. 
\end{proposition}
\begin{proof}
	Please refer to \bf Appendix C.
\end{proof}

Moreover, comparing Proportions~\ref{main-lobe rang} and~\ref{LSA GL Beam-depth}, we have
\begin{equation}
	\begin{aligned}
		\frac{\mathrm{BD}_{{{\rm LSA},u}}^{\rm (G)}}{\mathrm{BD}_{{\rm LSA}}^{\rm (M)}} & \!=\! \frac{1\!-\!(\sin\theta_{0}\!+\!\Delta_{{\rm LSA},u})^{2}}{\cos^{2}\theta_{0}}\!\\
		&=\! 1\!-\!\frac{\Delta_{{\rm LSA},u}(\Delta_{{\rm LSA},u}\!+\!2\sin\theta_{0})}{\cos^{2}\theta_{0}}, \forall u,
	\end{aligned}\label{NewEq26}	
\end{equation}
where $ \Delta_{{\rm LSA},u} \triangleq \sin\theta^{(\rm G)}_{{\rm LSA}, u} - \sin\theta_0$.
\begin{remark}[Beam-depths of main-lobe versus grating-lobes]\label{remark2} 	\rm From (\ref{NewEq26}), we can obtain the following useful insights.
	\begin{itemize}
		\item{When $\theta_{0} = 0$, the beam-depth of LSA grating-lobes is always smaller than that of its main-lobe, indicating that the grating-lobes are focused on a smaller region. Moreover, when the grating-lobes deviate farther from the main-lobe in the angular domain, they have narrower beam-depths as numerically shown in Fig. \ref{fig_5}.}
		\item{{When $\sin\theta_{0} > 0$, the grating-lobes have longer beam-depths than the main-lobe when their angles $ \sin\theta^{(\rm G)}_{{\rm LSA}, u} \in ( -\sin\theta_0, \sin\theta_0) $, and the beam-depth monotonically increases when $ \theta^{(\rm G)}_{{\rm LSA}, u} $ approaches $ 0 $. In contrast, the beam-widths of grating-lobes in other angles are smaller than that of the main-lobe.}} Similarly, we can obtain the result for the case $ \sin\theta_0 < 0 $ by using array symmetry.
	\end{itemize}
\end{remark}

The above LSA beam pattern analysis shows that although the LSA can enable near-field communications by enlarging the array aperture, it may potentially incur severe IUI since it generate many grating-lobes with focused high power.


\vspace{-0.4cm}
\subsection{Hybrid Beamforming Design for LSA}
In this subsection, we design efficient hybrid beamforming for LSAs. We consider the weighted sum-rate maximization problem, which is formulated as below\footnote{The design framework can be readily extended to optimizing other communication performance.}.
\begin{equation*}
	\begin{aligned}
		({\bf P1}):\!	\max\limits_{\mathbf{F}_{\rm LSA,\mathrm{A}}, \mathbf{F}_{\rm LSA, \mathrm{D}}} &\sum\limits_{k = 1}^{K} { w_k \log (1 +  {\rm SINR}_k)} \\
		\qquad\qquad\quad{\text{s}}{\text{.t}}{\rm{. }} \qquad~  &\left|\left[\mathbf{F}_{\rm LSA,\mathrm{A}}\right]_{q, k}\right|=1, \forall q\in\mathcal{Q}_{\rm LSA} , k\in\mathcal{K}\\
		&\left \| \mathbf{F}_\mathrm{\rm LSA,A}\mathbf{F}_\mathrm{\rm LSA,D}  \right \|^{2}_{F}   \le {P_{\max}},
	\end{aligned}
\end{equation*}
where $ P_{\max} $ represents the maximum transmit power of the LSA and $w_k$ is the user weight. However, it is difficult to directly solve Problem (P1) due to the constant-modulus constraint and the mutual coupling between $ \mathbf{F}_\mathrm{\rm LSA,A} $ and $ \mathbf{F}_\mathrm{\rm LSA,D} $.


To address this issue, we first observe from the above LSA beam pattern analysis that when applying the MRT beamforming design for a target user, there may exist strong IUI when other users are located at the focused regions of its grating-lobes. Thus, we cannot directly apply the low-complexity algorithm that first maximizes the received  power at individual users using the MRT-based analog beamforming and then minimizes the residual IUI by optimizing the digital beamforming. 

Therefore, for the LSA, one efficient hybrid beamforming method is by minimizing the difference between the optimal digital beamforming and the hybrid beamforming, by solving the following problem
\begin{equation*}
	\begin{aligned}
	({\bf P2}):\!	\min _{\mathbf{F}_{\rm LSA,\mathrm{A}}, \mathbf{F}_{\rm LSA},\mathrm{D}} & \left\|\mathbf{F}^{\mathrm{opt}}_{\rm LSA}-\mathbf{F}_{\rm LSA ,\mathrm{A}} \mathbf{F}_{\rm LSA,\mathrm{D}}\right\|_F^2 \\
	\text { s.t. }~~~~~ & \left|\left[\mathbf{F}_{\rm LSA,\mathrm{A}}\right]_{q, k}\right|=1, \forall q\in\mathcal{Q}_{\rm LSA} , k\in\mathcal{K} \\
	& \left\|\mathbf{F}_{\rm LSA,\mathrm{A}} \mathbf{F}_{\rm LSA, \mathrm{D}}\right\|_F^2 \leq P_{\max },
	\end{aligned}
\end{equation*}
where $\mathbf{F}^{\mathrm{opt}}_{\rm LSA} $ represents the optimal fully-digital precoding matrix for LSA. 
To solve Problem (P2), one can employ the alternating optimization (AO) method that iteratively optimizes one of $\mathbf{F}_{\rm LSA,\mathrm{A}} $ and $ \mathbf{F}_{\rm LSA,\mathrm{D}}$ with the other one being fixed, until the convergence is achieved.  The detailed algorithm is similar to that in \cite{yu2016alternating} and thus is omitted for brevity. It is worth noting that due to the small number of antennas, the computational complexity for the above proposed algorithm is much smaller than the ULA case given the same aperture size, which is thus affordable in practice.

\section{Extended Coprime Array}\label{ECA}
In this section, we first analyze the ECA beam pattern and then demonstrate its appealing capability to reduce the power of grating-lobes. Subsequently, we propose an efficient hybrid beamforming method for ECA. 

\begin{figure*}[ht]
	\begin{equation}
		{\small
			\begin{aligned}\label{ECAsum}
				f_{\rm ECA}  
	 			\approx\frac{1}{Q_{\rm ECA}}\left| 
				\sum_{m\in \mathcal{M}}e^{\left({\jmath {\frac{2\pi d_{0}N\Delta }{\lambda} m }+\jmath {\frac{\pi( N d_{0})^2\Phi}{\lambda} m^2  }}\right)}
				+\sum_{n\in \mathcal{N}}e^{\l(\jmath  {\frac{2\pi M d_{0}}{\lambda} n   \Delta}+\jmath {\frac{\pi ( M d_{0})^2\Phi}{\lambda} n^2  }\r)}
				-\sum_{i\in \mathcal{I}}e^{\l(\jmath  {\frac{2\pi  MN d_{0}}{\lambda} i  \Delta}+\jmath {\frac{\pi ( MN d_{0})^2\Phi}{\lambda} i^2  }\r)} \right|\triangleq \hat{f}_{\rm ECA}.
		\end{aligned}}
	\end{equation}
	\hrulefill
\end{figure*}


First, for an ECA, it can be easily obtained that its beam pattern follows
\vspace{-0.4cm}
\begin{align}\label{Eq:ECAf}
  &f_{\rm ECA}({r_0},{\theta _0};r,\theta ) \overset{(a_3)}{=} \left| \bar{\mathbf{b}}_{\rm ECA}^{H}\left(r_0,\theta_0\right) \times \Pi^{H}\Pi \times  
\bar{\mathbf{b}}_{\rm ECA}\left(r,\theta\right) \right| \nonumber\\
&=\left|\bar{\mathbf{b}}_{\rm ECA}^{H}\left(r_{0},\theta_{0}\right)  
\bar{\mathbf{b}}_{\rm ECA}\left(r,\theta\right)\right| \nn \\
&
= \left| \mathbf{b}_{1}^{H}\left(r_{0},\theta_{0}\right)\mathbf{b}_{1}\left(r,\theta\right)+
\mathbf{b}_{2}^{H}\left(r_{0},\theta_{0}\right)\mathbf{b}_{2}\left(r,\theta\right) \right|
\end{align}
where $(a_3)$ is due to $ \Pi^{H}\Pi =  \mathbf{I}$. Moreover, the beam pattern in \eqref{Eq:ECAf} can be further approximated as in \eqref{ECAsum} by means of Fresnel approximation.
	
Similar to the LSA, we divide the ECA beam pattern analysis into three cases according to the value of $ \Delta $ and $ \Phi $.
\begin{itemize}
\item {\bf Case 1:} $\Phi=0$.
	This case corresponds to the user-ring under the condition of $\frac{\cos^2 \theta}{r} =  \frac{\cos^2 \theta_0}{r_0}$. Its beam pattern is similar to the far-field ECA case, for which the beam-width for both the main-lobe and grating-lobes in the ring can be obtained.
	\item {\bf Case 2:} $\Delta = \frac{2 s}{ MN},  s\in \mathcal{S}$  where  $\mathcal{S} \triangleq \{\pm 1,\dots, \pm (MN-1)\}$.
	It will be shown that this condition is equivalent to specify several angles, for which we can obtain the beam-depth of the ECA grating-lobes. 
	\item {\bf Case 3:} $\Phi\neq 0$ and $\Delta\neq \frac{2 s}{ MN},  s\in \mathcal{S}$. For this case,  the expression of  $f_{\rm ECA}$ is still given in \eqref{ECAsum}. We  show below that in most cases, the ECA beam pattern $f_{\rm ECA}\approx 0$.

\end{itemize}

\begin{lemma}\label{Them4}\emph{
	 If $\Phi=0$  and $ \Delta_{\rm ECA} \neq \frac{2s}{MN}, s \in \mathcal{S}$  then $\hat{f}_{\rm ECA}({r_0},{\theta _0};r,\theta ) $  in (\ref{ECAsum}) can be approximated as
	 	\begin{align}	\label{ECAclosed-form}
	 		& \hat{f}_{\rm ECA}({r_0},{\theta _0};r,\theta ) \approx \left| \frac{LM-1}{Q_{\rm ECA}}G(\gamma_{1},  \gamma_{2})+ \right.\nn\\
	 		&\qquad\quad\quad\left.\frac{LN-1}{Q_{\rm ECA}}G(\gamma_{1}, \rho_{1}\gamma_{2})-\frac{L-1}{Q_{\rm ECA}}G(\gamma_{1}, \rho_{2}\gamma_{2}) \right|,\!\!\!
	 	\end{align}	
	 where $ \gamma_1 $ {\rm and} $ \gamma_2 $  are defined as
	\begin{equation}
		\gamma_{1} \triangleq \frac{\Delta}{{\sqrt{d_0\left| \Phi  \right|}}},
		~~
		\gamma_{2} \triangleq \frac{(LM-1)N}{2}\sqrt{{d_0}\left| \Phi \right| },
	\end{equation}
	 with  $ \rho_{1} \triangleq  \frac{LMN - M}{LMN - N}$ and $ \rho_{2} \triangleq   \frac{LMN - MN}{LMN - N} = \frac{LM - M}{LM - 1} $.}
\end{lemma}
\begin{proof}
	From Lemma \ref{LSAgeneral}, we have
	{\small
		\begin{equation}\label{EQ 37}
			\begin{aligned}
				{\bf{b}}_{{1}}^H(r,\theta ){{\bf{b}}_{1}}({r_0},{\theta _0})\approx
				\left\{\begin{matrix}
				 \frac{(LM-1)e^{-\jmath\pi\frac{\Delta^{2}}{2d\Phi}}}{2Q_{\rm ECA}} \frac{\widehat{C}(\gamma_{1},\gamma_{2}) +  \jmath(\widehat{S}(\gamma_{1},\gamma_{2})}{\gamma_2}, \Phi > 0,\\
				 \frac{(LM-1)e^{-\jmath\pi\frac{\Delta^{2}}{2d\Phi}}}{2Q_{\rm ECA}} \frac{\widehat{C}(\gamma_{1},\gamma_{2}) -  \jmath(\widehat{S}(\gamma_{1},\gamma_{2})}{\gamma_2}, \Phi < 0.
				\end{matrix}\right.
			\end{aligned}		
	\end{equation}}
Similarly, for subarray 2, $ {{\bf{b}}_{2}}({r_0},{\theta _0}){\bf{b}}_{{2}}^H(r,\theta ) $ is given by (\ref{SumSubarray2}), where $ \rho_{1} = \frac{LMN - M}{LMN - N} $ and $ \rho_{2} = \frac{LMN - MN}{LMN - N} $.  Combining  \eqref{EQ 37} and \eqref{SumSubarray2} leads to the desired result in (\ref{ECAclosed-form}).
\end{proof}
\begin{figure*}[ht] 
		\centering
		\begingroup\makeatletter\def\f@size{10}\check@mathfonts
		\begin{equation}
			{{\bf{b}}_{{2}}^H(r,\theta ){{\bf{b}}_{2}}({r_0},{\theta _0})\approx}
			\begin{aligned}
				\left\{\begin{matrix}
					\frac{(LN-1)e^{-\jmath\pi\frac{\Delta^{2}}{2d\Phi}}}{2Q_{\rm ECA}} \frac{\widehat{C}^{}(\gamma_{1},\rho_{1}\gamma_{2}) + \jmath(\widehat{S}(\gamma_{1},\rho_{1}\gamma_{2})}{\gamma_2} - \frac{(L-1)e^{-\jmath\pi\frac{\Delta^{2}}{2d\Phi}}}{2Q_{\rm ECA}} \frac{\widehat{C}(\gamma_{1},\rho_{2}\gamma_{2}) +  \jmath(\widehat{S}(\gamma_{1},\rho_{2}\gamma_{2})}{\gamma_2}, \Phi> 0
					\\
					\frac{(L-1)e^{-\jmath\pi\frac{\Delta^{2}}{2d\Phi}}}{2Q_{\rm ECA}} \frac{\widehat{C}(\gamma_{1},\rho_{1}\gamma_{2}) -  \jmath(\widehat{S}(\gamma_{1},\rho_{1}\gamma_{2})}{\gamma_2} -  \frac{(L-1)e^{-\jmath\pi\frac{\Delta^{2}}{2d\Phi}}}{2Q_{\rm ECA}} \frac{\widehat{C}(\gamma_{1},\rho_{2}\gamma_{2}) -  \jmath(\widehat{S}(\gamma_{1},\rho_{2}\gamma_{2})}{\gamma_2}, \Phi< 0
				\end{matrix}\right.
			\end{aligned}
			\label{SumSubarray2}
		\end{equation}
		\endgroup\hrulefill
	\end{figure*}

Similarly, by using  the properties of function $ G(\cdot) $, we can easily show that except for the cases of $ \Delta \to 0 $, the ECA beam pattern nearly approaches to zero.
\vspace{-0.4cm}
\subsection{Main-lobe of ECA Beam Pattern}

In this subsection, we  characterize the beam-width and beam-depth of the ECA main-lobe.
\subsubsection{{Beam-width in the Angular Domain}}
To analyze the beam-width of the ECA main-lobe, we consider the user-ring $ \frac{\cos^2 \theta}{r} =  \frac{\cos^2 \theta_0}{r_0} $ and then have the following result.
\begin{lemma}\label{ECAgeneralcase}\emph{
	When $\Phi=0$ (or equivalently $ \frac{\cos^2 \theta}{r} =  \frac{\cos^2 \theta_0}{r_0}$),  $\hat{f}_{\rm ECA}({r_0},{\theta _0};r,\theta ) $ in (\ref{ECAsum}) is given by
	\begin{equation}\label{eq52}
		{\small
			\begin{aligned}
				\hat{f}_{\rm ECA} =
				 \frac{\left| {\Xi_{LM-1}({\Delta}N)} + \right.
				 	\left.	{\Xi_{LN-1}(\Delta M)}-{\Xi_{L-1}(\Delta MN)} \right|}{{Q_{\rm ECA}}}.
			\end{aligned}	}	
	\end{equation}}
\end{lemma}
\begin{proof}
	{\rm If} $ \frac{\cos^2 \theta}{r} =  \frac{\cos^2 \theta_0}{r_0} $, based on Lemma \ref{lemma3}, the beam pattern of subarray $1$ is given by
	\begin{equation}
		\begingroup\makeatletter\def\f@size{9}\check@mathfonts
		\begin{aligned}
			{\bf{b}}_{1}^{\rm{H}}({r_0},{\theta _0}){{\bf{b}}}_{1}(r,\theta )= \frac{{e^{\jmath\frac{{2\pi }}{\lambda }({r_0} - r)}}\Xi_{LM-1}(\Delta N)}{Q_{\rm ECA}},
		\end{aligned}\label{eq42}
		\endgroup
	\end{equation}
	Similarly, for subarray 2, we have
	{
		\begin{equation}
			\begin{array}{l}
				{\bf{b}}_{2}^{H}({{{r}}_0},{{\bf{\theta }}_0}){{\bf{b}}_{2}}({r},{\bf{\theta }}) = \frac{{e^{\jmath\frac{{2\pi }}{\lambda }({r_0} - r)}}\left( \Xi_{LN-1}(\Delta M)-\Xi_{L-1}(\Delta MN) \right)}{Q_{\rm ECA}}.
			\end{array}\nn
	\end{equation}}
\vspace{-0.2cm}
Combining the above leads to the desired result.
\end{proof}

\begin{proposition}[Beam-width of ECA main-lobe]
	\emph{For an ECA parameterized by $\{M, N, L\}$, the main-lobe beam-width of the ECA beamforming vector $ {{\bf{b}}}_{\rm ECA}(r_{0},\theta_{0} ) $ can be expressed as}
	\begin{equation}\label{ECA beam-width}
		\mathrm{BW}_{{\rm ECA}}^{\rm (M)} \approx \frac{4}{(LM-1)N} \approx \frac{4d_0}{A_{\rm ECA}}.
	\end{equation}
\end{proposition}
\begin{proof}
	From (\ref{eq52}), we can obtain that the main-lobe beam-widths of subarray 1 and subarray 2 are  $ \mathrm{BW}_{{\rm ECA}}^{\rm (M,1)} = \frac{4}{(LM-1)N} $ and 
	$ \mathrm{BW}_{{\rm ECA}}^{\rm (M,2)} = \frac{4}{(LN-1)M} $ at $ \theta =  \theta_0 $, respectively. Then, by using $ LMN-M \approx LMN-N $,  the ECA beam-width is obtained as
	$ \mathrm{BW}_{{\rm ECA}}^{\rm (M)} \approx \frac{4}{(LM-1)N} \approx \frac{4d_0}{A_{\rm ECA}}. $
\end{proof}

Similar to the LSA, (\ref{ECA beam-width}) indicates that the beam-width of the ECA  main-lobe  is {\it inversely} proportional to the array aperture. Moreover, comparing \eqref{ECA beam-width} and \eqref{LSA beam-width}, it shows that when the LSA and ECA have the  (approximately) same sparsity and number of antennas, they have the same main-lobe beam-width, as illustrated in Figs. \ref{fig_4}(b) and \ref{fig_4}(c).
\subsubsection{Beam-depth in the Range Domain}
To characterize the beam-depth of the LSA main-lobe, we set $ \theta = \theta_{0} $.
Based on Lemma \ref{ECAgeneralcase} with $ \theta = \theta_{0} $ (equivalently $ \gamma_{1} = 0$),  the ECA beam pattern can be approximated as follows.
\begin{lemma}\label{ECA range domain}
	\emph{When $ \theta = \theta_{0} $, the ECA beam pattern in (\ref{ECAclosed-form}) reduces to
		\begin{equation}{\small
				\begin{aligned}\label{eQ48}
					\!\!\!\!\!\!\hat{f}_{\rm ECA} \!=\!\! 
					\left| \frac{LM-1}{Q_{\rm ECA}}F(\gamma_{2}) \!+\! \frac{LN-1}{Q_{\rm ECA}}F(\rho_{1}\gamma_{2})\!\!-\!\frac{L-1}{Q_{\rm ECA}}F( \rho_{2}\gamma_{2}) \right|.\!\!\\
			\end{aligned}}		
		\end{equation}
	}
\end{lemma}
Although in closed-form, it is still difficulty to obtain useful insights from \eqref{eQ48}. As such, we consider a typical scenario with $ M \approx N $ and $ L \ge 4 $ (which will be shown in Section \ref{ECAGratingLobe} to effectively mitigate the ECA grating-lobes), for which we have $ \rho_1 =  \frac{LMN - M}{LMN - N} \approx 1$ and  $ \rho_2 = \frac{LM - M}{LM - 1}  \ge \frac{L-1}{L} \ge 0.75$. As such,  $ f_{\rm ECA}  $ in \eqref{eQ48} can be approximated as 
\begin{equation}\label{simplified 37}
	\hat{f}_{\rm ECA}({r_0},{\theta _0};r,\theta ) \approx \left| F(\gamma_{2}) \right|.
\end{equation}
In Fig.~\ref{fig_7}, we plot the curves of \eqref{eQ48} and \eqref{simplified 37} given $ \rho_2 = 0.75 $, which verifies the accuracy of  made approximation in \eqref{simplified 37}.

\begin{figure}[t]
	\begin{center}
		\includegraphics[width=2.7in]{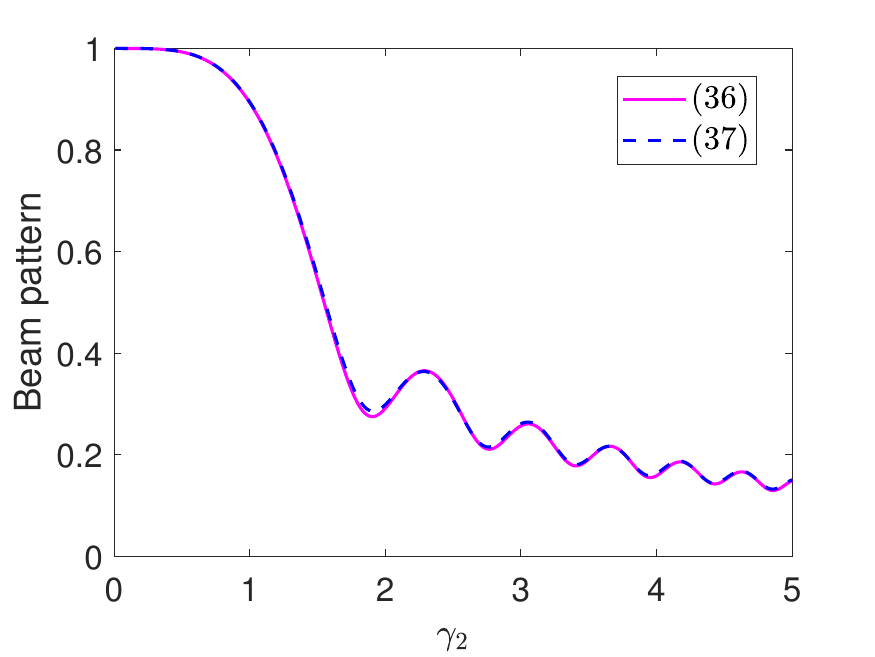}
		\caption{{Accuracy of the made approximation in \eqref{simplified 37}.}} 
		\label{fig_7}
		\vspace{-0.5cm} 
	\end{center}	
\end{figure}

Next, we obtain the beam-depth of ECA main-lobe as below.
\begin{proposition}[Beam-depth of ECA main-lobe]\label{main-lobe range ECA}
	\rm For an ECA parameterized by $\{M, N, L\}$, the 3-dB  beam-depth for the main-lobe of ECA beamformer ${\bf{b}}_{\rm ECA}({r_0},{\theta_0 })$ is given by
\vspace{-3pt}
	\begin{equation}
		\mathrm{BD}_{{\rm ECA}}^{\rm (M)}= \begin{cases}\frac{2 r_{0}^2 r_{\mathrm{ECA}}}{r_{\mathrm{ECA}}^2-r_{0}^2}, & r_{0}<r_{\mathrm{ECA}} \\ \infty, & r_{0} \geq r_{\mathrm{ECA}}\end{cases},
	\end{equation}
	\rm where $ r_{\mathrm{ECA}} \approx \frac{(LM-1)^{2}N^{2}d_0 \cos ^2 \theta_{0}}{4  \varphi _{3 \mathrm{dB}}^2}$. 
\end{proposition}
\begin{proof}
	The proof is similar to that of Proposition \ref{main-lobe rang} and  thus we omit it for brevity.
\end{proof}

 By comparing Propositions \ref{main-lobe rang} and \ref{main-lobe range ECA}, it indicates that the two SA configurations have the same beam-depth when they have the same aperture, which are shown in Figs. \ref{fig_4}(b) and \ref{fig_4}(c).
 
 \vspace{-0.5cm}
\subsection{Grating-lobes of ECA Beam Pattern}\label{ECAGratingLobe}
In this subsection, we characterize the beam-width, beam-depth, and beam-height of the ECA grating-lobes. More importantly, we reveal that the ECA can not only strengthen the beam-focusing effect, but also significantly suppress the magnitude of grating lobes.
\subsubsection{Beam-width in the Angular Domain}
Based on Lemma \ref{ECAgeneralcase}, the locations and heights of ECA grating-lobes are obtained.

\begin{proposition}\label{GrateLobeDirection}
	\rm For an ECA parameterized by $\{M, N, L\}$, it has three types of grating-lobes. Specifically, the {Type-I and Type-II  grating-lobes} occur at the angles of  
	$$ \theta^{(\rm G,1)}_{{\rm ECA}, n} \!=\! \arcsin\l(\sin\theta_0\!+\!\frac{2n}{ N}\r), \forall n\in \mathcal{N},$$  $$ \theta^{(\rm G,2)}_{{\rm ECA}, m} \!=\! \arcsin\l(\sin\theta_0\!+\!\frac{2m}{ M}\r), \forall m\in \mathcal{M},$$ 
	respectively. 
	Moreover, the Type-III grating-lobes occur at the angles of 
	$ \theta^{(\rm G,3)}_{{\rm ECA}, \ell} \!=\! \arcsin\l(\sin\theta_0\!+\!\frac{2\ell}{ MN}\r), \forall \ell\in \mathcal{L}$.
\end{proposition}
\begin{proof}
	Considering $\lim\limits_{x \to k\pi} \frac{\sin Qx}{\sin x} = Q$, the three terms of (\ref{eq52}) respectively generate grating-lobes at {$ \Delta  = \frac{2n}{N}, n = \in \mathcal{N}$}, {$ \Delta = \frac{2m}{M}, m \in  \mathcal{M}$} and { $ \Delta  = \frac{2s}{MN}, s \in \mathcal{S} $}.
\end{proof}

\begin{figure}[t]
	\begin{center}
		\includegraphics[width=2.7in]{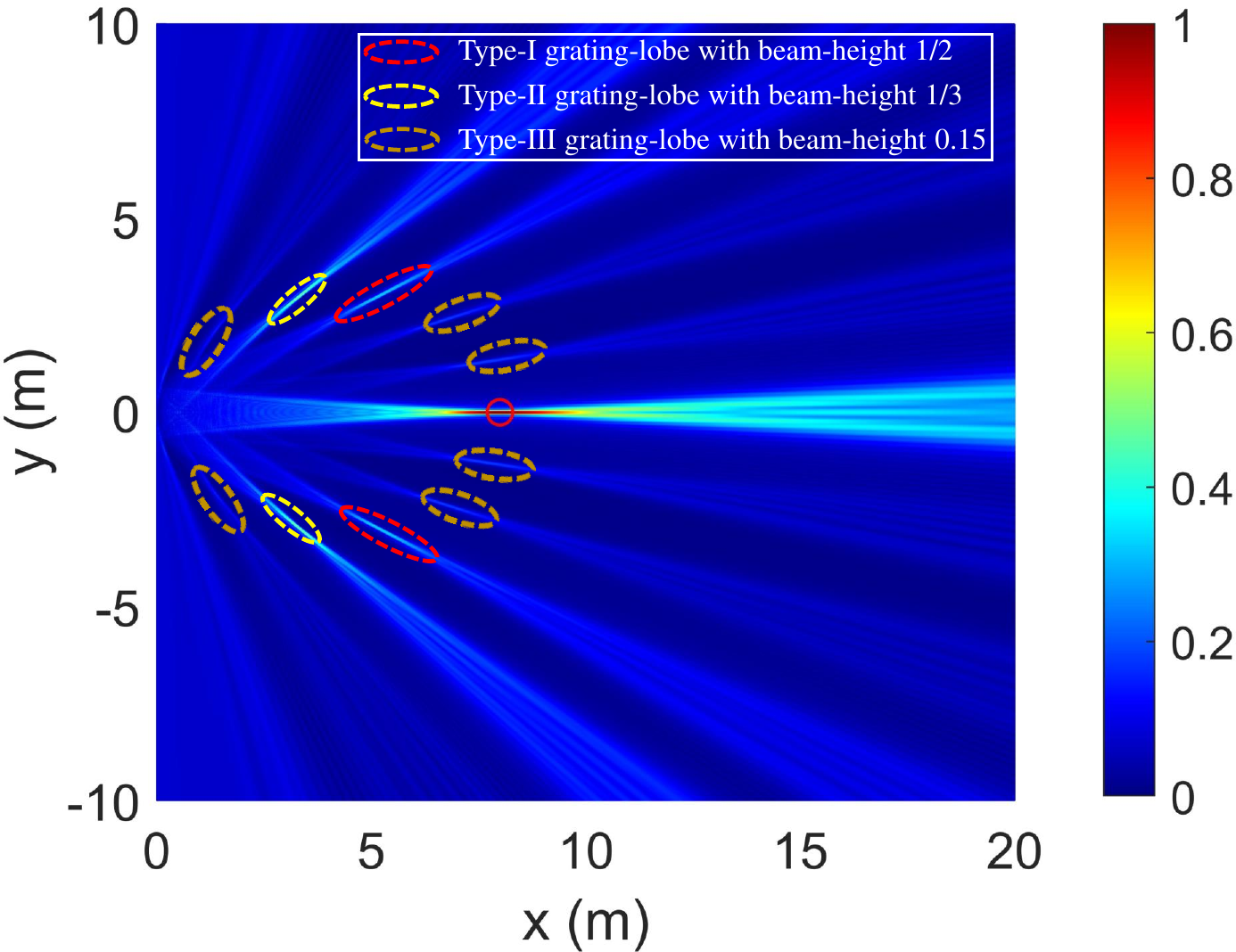}
		\caption{{Different types of grating-lobes with $ M = 4 $, $ N = 3 $ and $ Q_{\rm ECA} = 131 $} under $ \lambda = 0.01$ m at ($ \theta = 0^{\circ} $, $ r = 8 $ m)} 
		\label{fig_8}
	\end{center}
\end{figure}

For an ECA, we call its grating-lobes as the {\it {Type-I}} and {\it Type-II}, when it satisfies $\l\{\Delta = \frac{2 n}{ N},  n\in \mathcal{N}\r\}$ and $\l\{\Delta = \frac{2 m}{ M},  m\in \mathcal{M}\r\}$, respectively. Besides, when $\l\{\Delta = \frac{2 \ell}{ MN},  \ell\in \mathcal{L}\r\}$ with $ \mathcal{L} =  \mathcal{S}\setminus\mathcal{M}\setminus\mathcal{N}$, we term the grating-lobes generated by the common antennas shared by the two LSAs as the {\it Type-III} grating-lobes, which will be shown to have much smaller beam-heights than  Type-I and Type-II grating-lobes. The three types of grating-lobes are illustrated in Fig. \ref{fig_8}. 
\begin{proposition}\label{NotOverlap}
	\rm The positions of Type-I and Type-II grating-lobes generated by the two subarrays do not overlap.
\end{proposition}
\begin{proof}
	If the grating-lobe positions generated by the two subarrays are the same, the following condition needs to be met: $\frac{2n}{N} = \frac{2m}{M}, m\in \mathcal{M}, n\in\mathcal{N}$, i.e., $ mN = nM $. 
%
	Due to the coprimality of $M$ and $N$ with their smallest common multiple being $ MN $, the above condition cannot be satisfied, thus leading to the desired result.
\end{proof}

{Comparing Propositions \ref{GrateLobeDirection} and \ref{NotOverlap},  {there are $2(N-1)$ Type-I grating-lobes, $ 2(M-1) $ Type-II grating-lobes and $ 2(MN-(M+N-1)) $ Type-III grating-lobes in the whole space, respectively.}}
Additionally, similar to ECA main-lobe, the beam-width of Type-I and Type-II grating-lobes can be easily obtained as {\small $ \mathrm{BW}_{{\rm ECA}}^{\rm (G,1)} {\approx} \mathrm{BW}_{{\rm ECA}}^{\rm (G,2)} {\approx} \frac{4d_{0}}{A_{\rm ECA}} $}. Moreover, the beam-width of Type-III grating-lobes is given by {\small $ \mathrm{BW}_{{\rm ECA}}^{\rm (G,3)} = \frac{4}{(L-1)MN} $}.

Furthermore, to avoid the beam-width overlapping in the Type-I and Type-II grating-lobes that may potential incur high IUI otherwise, we impose the following constraint on the number of ECA antennas.
Note that the minimum spatial angular separation $ \Delta_{\rm min} $ between Type-I and Type-II grating-lobes is given by 
\begin{equation}
	\begin{aligned}
		\Delta_{\rm min} = \min\limits_{m\in \mathcal{M}, n\in\mathcal{N}}& \l\{ \left|\frac{2n}{N} - \frac{2m}{M}\right|\r\} \ge \frac{2}{MN}.
	\end{aligned}	
\end{equation}
To avoid the beam pattern overlapping,  we need the condition  $ \frac{2}{MN} > \mathrm{BW}_{{\rm ECA}}^{\rm (G,1)} $, which is equivalent to $L >  \frac{2MN+ N}{MN}=2+ \frac{1}{M}$. For the even integer $L$, we thus have $L >4$. 
\begin{proposition}\label{ProBeam-height-ECA}
	\rm For an ECA parameterized by $\{M, N, L\}$, the beam-heights of its Type-I and Type-II grating-lobes are
	{\small
	$${\rm BH}^{(1)}_{n} = \frac{M-1}{M+N-1},\forall n\in \mathcal{N},~~ {\rm BH}^{(2)}_{m} = \frac{N-1}{M+N-1},\forall m\in \mathcal{M},$$}respectively. 
	Moreover, the beam-heights of Type-III grating-lobes are given by $ {\rm BH}^{\rm (3)}_{\ell} = \frac{L-1}{L(M+N-1)-1}, \forall \ell \in \mathcal{L} $.
\end{proposition}
\begin{proof}
	Given $ L \ge 4 $ such that the Type-I and Type-II grating-lobes does not overlap in the angular domain, when the first term of \eqref{eq52} generates Type-I grating-lobes, 
	it is easy to show that  
	\begin{equation}
		\label{LSAbeamHeight}
		{\small 
		\begin{aligned}
			{\rm BH}^{(1)}_{n}=  \frac{LM-1}{Q_{\rm ECA}} - \frac{L-1}{Q_{\rm ECA}} \approx \frac{M-1}{M+N-1}, \forall n.
		\end{aligned}}
	\end{equation}
	Similarly to $ {\rm BH}^{(1)}_{n}$, we can obtain $ {\rm BH}^{(2)}_{m} =  \frac{LN-1}{Q_{\rm ECA}} - \frac{L-1}{Q_{\rm ECA}} \approx \frac{N-1}{M+N-1}, \forall m  $ and $ {\rm BH}^{(3)}_{\ell}=  \frac{L-1}{Q_{\rm ECA}}, \forall \ell $.

\end{proof}
\begin{remark}[Selection of $M$ and $N$]	
	\rm To suppress the grating-lobes generated by the two effective LSAs, it can be observed from \eqref{LSAbeamHeight} that it is desirable to set $ M \approx  N $. Then, the following approximation can be made
	\begin{equation}
		\begin{aligned}
			{\rm BH}^{(1)}_{n}=  \frac{M-1}{M+N-1}\approx\frac{1}{2}, \forall m.
		\end{aligned}
	\end{equation}
	\rm
Otherwise, if $ M \gg N $, we have $ {\rm BH}^{(1)}_{n}\approx 1 $ and $ {\rm BH}^{(2)}_{m}\approx 0 $, which reduces to the LSA case and may potentially degrade the user performance when it is located at the grating-lobes of other user's beamformer. 	
\end{remark}
\vspace{-0.2 cm}
\begin{remark}[Grating-lobes smoothing]\emph{Compared with the LSA, the ECA beam pattern generally have more grating-lobes under the same array sparsity. However, it is also worth noting that the beam-heights of ECA grating-lobes are much smaller than those of LSA, thanks to the offset effect of the grating-lobes generated by the two subarrays. This phenomenon is termed as the \emph{grating-lobes smoothing} in this paper. It thus leads to the better worst-case user-rate performance in the ECA communication system as compared to the LSA counterpart. 
	}
\end{remark}

\subsubsection{Beam-depth in the Range Domain}
For the Type-I, Type-II, and Type-III grating-lobes,  we set $ \theta = \theta^{(\rm G,1)}_{{\rm ECA}, n} $, $ \theta = \theta^{(\rm G,2)}_{{\rm ECA}, m} $ and $ \theta = \theta^{(\rm G,3)}_{{\rm ECA}, \ell} $, respectively. Then, we have the following result on their beam-depths.
\begin{lemma}[]\label{lemma11}\emph{
	{\rm If}  $ \Delta = \frac{2m}{M}, m \in  \mathcal{M} $, $ \Delta = \frac{2n}{N}, n \in \mathcal{N}$ or $ \Delta = \frac{2\ell}{MN}, \ell \in \mathcal{L}$, then {\rm $ \hat{f}_{\rm ECA}$ in (\ref{ECAsum})  can be respectively approximated as}	
	\begin{equation}\label{newEQ43}
			{\small\begin{aligned}
				\hat{f}_{\rm ECA}({r_0},{\theta _0};r,\theta^{(\rm G,1)}_{{\rm LSA},n} ) \approx 
				\left| \frac{LM-1}{Q_{\rm ECA}}F(\gamma_{2})-\frac{L-1}{Q_{\rm ECA}}F( \rho_{2}\gamma_{2}) \right|, \forall n, 
			\end{aligned}}	
	\end{equation}
	 \begin{equation}\label{newEQ44}
			{\small\begin{aligned}
				\hat{f}_{\rm ECA}({r_0},{\theta _0},r,\theta^{\rm (G,2)}_{{\rm LSA},m} )\!\! \approx \!\!
				\left| \frac{LN-1}{Q_{\rm ECA}}F(\rho_{1}\gamma_{2})-\frac{L-1}{Q_{\rm ECA}}F( \rho_{2}\gamma_{2}) \right|, \forall m,
			\end{aligned}}	
	\end{equation}
	\begin{equation}\label{newEQ45}
		{\small \begin{aligned}
			\hat{f}_{\rm ECA}({r_0},{\theta _0},r,\theta^{(\rm G,3)}_{{\rm LSA},\ell} ) \approx 
			\left| \frac{L-1}{Q_{\rm ECA}}F( \rho_{2}\gamma_{2}) \right|, \forall \ell.
		\end{aligned}}	
	\end{equation}
}
\end{lemma}
\begin{proof}
	When $ \Delta = \frac{2m}{M}, m \in  \mathcal{M} $, for subarray 1, similar to Appendix A, its beam pattern can be obtained as 
	{\small
	\begin{equation}		
		{\bf{b}}_{{1}}^H(r,\theta^{(\rm G,2)}_{{\rm ECA}, m} ){{\bf{b}}_{1}}({r_0},{\theta _0})\approx
		\begin{aligned}
			\left\{\begin{matrix}
				\frac{(LN-1)e^{-\jmath\pi\frac{\Delta^{2}}{d\Phi}}}{2Q_{\rm ECA}} G(\gamma_1, \gamma_2), \Phi> 0
				\\
				\frac{(LN-1)e^{-\jmath\pi\frac{\Delta^{2}}{d\Phi}}}{2Q_{\rm ECA}} G^{\ast}(\gamma_1, \gamma_2) , \Phi < 0
			\end{matrix}\right.
		\end{aligned}
	\end{equation}}
Similarly, for subarray 2, we have
		\begin{equation}\label{EQ55}
		{\small
			{\bf{b}}_{{2}}^H(r,\theta^{(\rm G,2)}_{{\rm ECA}, m} ){{\bf{b}}_{2}}({r_0},{\theta _0}) =\frac{LN-1}{Q_{\rm ECA}}F(\rho_{1}\gamma_{2}) - \frac{L-1}{Q_{\rm ECA}}F(\rho_{2}\gamma_{2}).}
	\end{equation}
Given $ \Delta = \frac{2m}{M}, m \in \mathcal{M} $ and using the observation in Fig. \ref{fig3}, we have {$ {\bf{b}}_{{1}}^H(r,\theta^{(\rm G,2)}_{{\rm ECA}, m} ){{\bf{b}}_{1}}({r_0},{\theta _0})\!\approx\! 0 $}. Then, we have $ |{\bf{b}}_{{\rm{ECA}}}^H(r,\theta^{(\rm G,2)}_{{\rm ECA}, m} ){{\bf{b}}_{{\rm{ECA}}}}({r_0},{\theta _0})|
	\approx |{\bf{b}}_{{2}}^{H}(r,\theta^{(\rm G,2)}_{{\rm ECA}, m} ){{\bf{b}}_{2}}(r_{0},\theta_{0} ) |$ and thus obtain the desired results in \eqref{newEQ44}.
	By using the similar method as for the case  $ \Delta = \frac{2m}{M}, m \in  \mathcal{M} $, we can obtain the result in \eqref{newEQ43} and \eqref{newEQ45}, thus completing the proof.
\end{proof}	

To gain further insights, similar to the rang domain of the main-lobe, \eqref{newEQ43} and \eqref{newEQ44} can be simplified as
{
	\begin{equation}
		\hat{f}_{\rm ECA}({r_0},{\theta _0},r,\theta^{\rm (G,1)}_{{\rm LSA},n} ) \approx 
		\left| \frac{LN-L}{Q_{\rm ECA}}F(\gamma_{2}) \right|, \forall n,
\end{equation}}
{
	\begin{equation}
		\hat{f}_{\rm ECA}({r_0},{\theta _0},r,\theta^{\rm (G,2)}_{{\rm LSA},m} ) \approx 
		\left| \frac{LM-L}{Q_{\rm ECA}}F(\gamma_{2}) \right|, \forall m.
		\vspace{-5pt}
\end{equation}}

For the beam-depth of the ECA grating-lobes, we present the main result below for the Type-I grating-lobes, while the results for the Type-II and Type-III grating-lobes can be obtained by using the similar methods. The proof is similar to Proposition 3 and thus omitted due to limited space.

\begin{proposition}[Beam-depth of Type-I grating-lobe]\label{ECA GL Beam-depth}
	\rm For an ECA parameterized by $\{M, N, L\}$, the beam-depths of Type-I grating-lobes are given by
	\begin{equation}
		\mathrm{BD}_{{\rm ECA}, n}^{\rm (G,1)} = \begin{cases}\frac{2 r_{0}^2 r_{{\mathrm{ECA},n}}^{\rm(G,1)}}{r_{\mathrm{ECA}}^2-r_{0}^2}, & r_{0}<r_{\mathrm{ECA}} \\ \infty, & r_{0} \geq r_{\mathrm{ECA}}\end{cases}, \forall n,
	\end{equation}
	\rm where $ r_{\mathrm{ECA},n}^{\rm(G,1)} \approx \frac{(LM-1)^{2}N^{2}d_0 \cos ^2 \theta^{\rm (G,1)}_{{\rm ECA},n}}{4  \varphi _{3 \mathrm{dB}}^2}, \forall n$. 
\end{proposition}

Hence, we have
\begin{equation}
	\begin{aligned}
		\frac{\mathrm{BD}_{{\rm ECA}, n}^{\rm (G,1)}}{\mathrm{BD}_{{\rm ECA}}^{\rm (M)}} &= \frac{\cos^{2}\theta^{\rm (G,1)}_{{\rm ECA},n}}{\cos^{2}\theta_{0}} = \frac{1-(\sin\theta_{0}+\Delta^{(1)}_{{\rm ECA},n})^{2}}{\cos^{2}\theta_{0}}\\
		&= 1-\frac{\Delta^{(1)}_{{\rm ECA},n}(\Delta^{(1)}_{{\rm ECA},n}+2\sin\theta_{0})}{\cos^{2}\theta_{0}},
	\end{aligned}\label{NewEq50}	
			\vspace{-5pt}
\end{equation} 
where $ \Delta^{(1)}_{{\rm ECA}, n} \triangleq\sin\theta^{({\rm G},1)}_{{\rm ECA}, n} - \sin\theta_0 $.

As \eqref{NewEq26} and \eqref{NewEq50} have similar forms,  the beam-depths of ECA grating-lobes have similar properties with the LSA counterparts, as discussed in Remark~\ref{remark2}.

	
\subsection{Hybrid Beamforming Design for ECA}\label{hybridECA}
In this subsection, we propose a customized hybrid beamforming design for ECAs. The corresponding weighted sum-rate maximization problem can be formulated similarly as that in Problem (P1) by replacing the LSA by the ECA. 

To solve this problem, instead of using the solution method proposed for the  LSA, we devise an alternative method by more effectively exploiting the grating-lobes smoothing effect for LSAs. Specifically, it is shown in Section~\ref{ECA} that compared with the LSA, the ECA substantially suppresses the interference power in the possible grating-lobes, thus improving the worst-case user rate performance. Motivated by the above, we propose an efficient and low-complexity hybrid beamforming design for the ECA, where the analog beamforming is devised to maximize the received signal power at each individual user and the digital beamforming is designed to cancel the residual inter-user interference \cite{liu2023near}.

Specifically, the analog beamformer for each user is given by $ \mathbf{f}_{{\rm ECA},\mathrm{A}, k} = \mathbf{b}_{k}(r_{k}, \theta_k)$. Then, the effect channel of user $k$ is  $ \tilde{\mathbf{h}}_k=\mathbf{F}_{\rm ECA,\mathrm{A}}^{H} \mathbf{h}_k $. As such, the SINR for user $k$, ${\rm SINR}_k$,  can be recast as $ {\rm SINR}_k=\frac{\left|\tilde{\mathbf{h}}_k^{H} \mathbf{f}_{{\rm ECA},\mathrm{D}, k}\right|^2}{\sum\limits_{i=1,i \neq k}^{K}\left|\tilde{\mathbf{h}}_k^{H} \mathbf{f}_{{\rm ECA},\mathrm{D}, i}\right|^2+\sigma_k^2}, \forall k.$
%
Therefore, we have the following optimization problem
\begin{equation*}
	\begin{aligned}
		({\bf P3}):\!	&\max\limits_{\mathbf{F}_{{\rm ECA},\mathrm{D}}} C = w_k\sum\limits_{k = 1}^{K} {\log (1 +  {\rm SINR}_{k})} \\
		&\quad{\text{s}}{\text{.t}}{\rm{. }} \quad  
		\left \| \mathbf{F}_\mathrm{\rm ECA,A}\mathbf{F}_\mathrm{\rm ECA,D}  \right \|^{2}_{F}   \le {P_{\max}}.
	\end{aligned}
\end{equation*}
To solve Problem (P3), we can employ the fractional programming (FP) algorithm detailed in \cite{ref5}, with the details omitted for brevity. 

\vspace{-0.1cm}
\section{Numerical Results}
In this section, we present numerical results to demonstrate the performance gain of the proposed SAs over conventional ULAs (i.e., DAs). 
The system parameters are set as follows. The BS operation at $ f = 30 $ GHz servers $ K = 10 $ users, with the reference path-loss $\beta_0 = (\lambda/4\pi)^{2} = -62 $ dB. The transmit SNR is $ \mathrm{SNR}_{\rm t} = \frac{P_{\rm max}}{\sigma^{2}} = 100$ dB with  $ P_{\rm max} = 30 $ dBm and the noise power $ \sigma^{2} = -70 $ dBm. Without specified otherwise, we consider three array configurations: 1) the LSA with $ U=3 $ and $ Q_{\rm LSA}=131 $; 2) the ECA with $ M = 7 $, $ N = 5 $ and $ Q_{\rm ECA}=131 $; and 3) the ULA with $ Q_{\rm ULA} = 131 $. Other parameters are set as: $ w_k=1, \forall k $. Three typical scenarios for the user distribution are considered as follows.

\vspace{-0.2cm}
\subsection{Scenario 1: Users at the same Sector Region}
Firstly, we consider the scenario where the users are uniformly and densely distributed within the  region $ \theta \in [-4^{\circ} , 4^{\circ}], r\in [20 ~{\rm m}, 30~ {\rm m}]$. In Fig. \ref{fig_9}, we plot the curves of the achievable sum-rate versus the array sparsity for different arrays under the (approximately) same antenna number.  
Note that it is not always possible to guarantee the same sparsity and number of antennas for the LSA and ECA due to their different antenna-spacing. For example, in Fig. \ref{fig_9}, the LSA and ECA have the same antenna number only when $ \eta = 2$ and $ \eta = 6$, while the ECA has a slightly larger number of antennas than the LSA for other values of $ \eta = 2$.  Moreover, for performance comparison, we consider the benchmark scheme with the ULA, for which the digital beamforming method is the same as that of SAs and the hybrid beamforming proposed in Section \ref{hybridECA} is adopted.
Based on the above, several important observations are made as follows. 
First,  the sum-rates of both SA configurations are significantly higher than that of ULA. 
This can be explained by two facts: 1) the SAs substantially enlarges the LSA aperture and thus endows the near-field beam-focusing effect, which is not achievable at the ULA with a small number of antennas; and 2) the beam-widths of SAs are much smaller than that of ULA, thus greatly reducing the IUI when users are randomly distributed. Second, the proposed hybrid beamforming design for the two SAs achieve very close rate performance with that of optimal digital beamforming, while the hybrid beamforming for ULA suffers considerable rate performance loss as compared to its digital beamforming. This can be intuitively understood, since the users in the ULA communication systems are at the  far-field region, thus the MRT-based analog beamforming cannot address the IUI properly. Third, it is observed that the ECA  system attains a bit higher rate than the LSA system at the sparsity of $\eta\in\{3,4,5\}$. This is mainly because at these cases, the ECA has a larger number of antennas and hence higher beamforming gain, while their rate performance at almost the same at $\eta=2$ and $\eta=6$, for which both SAs have the same antenna number. 
\begin{figure}[!t]
	\begin{center}
		\includegraphics[width=2.7in]{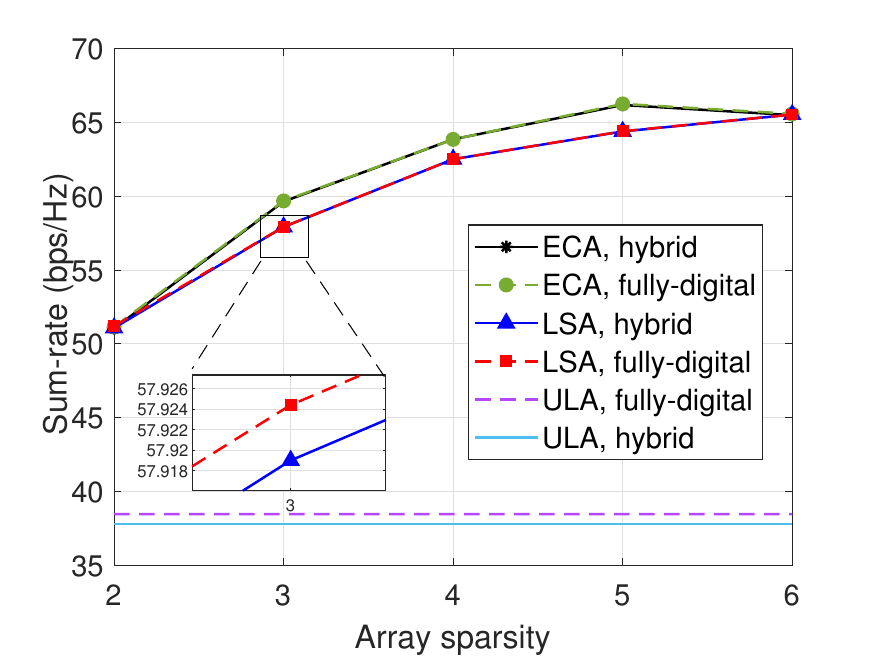}
		\caption{Scenario $1$: achievable sum-rate versus array sparsity.}
		\label{fig_9}
		\vspace{-0.6cm} 
	\end{center}	
\end{figure}
\begin{figure}[!t]
	\begin{center}
		\includegraphics[width=2.7in]{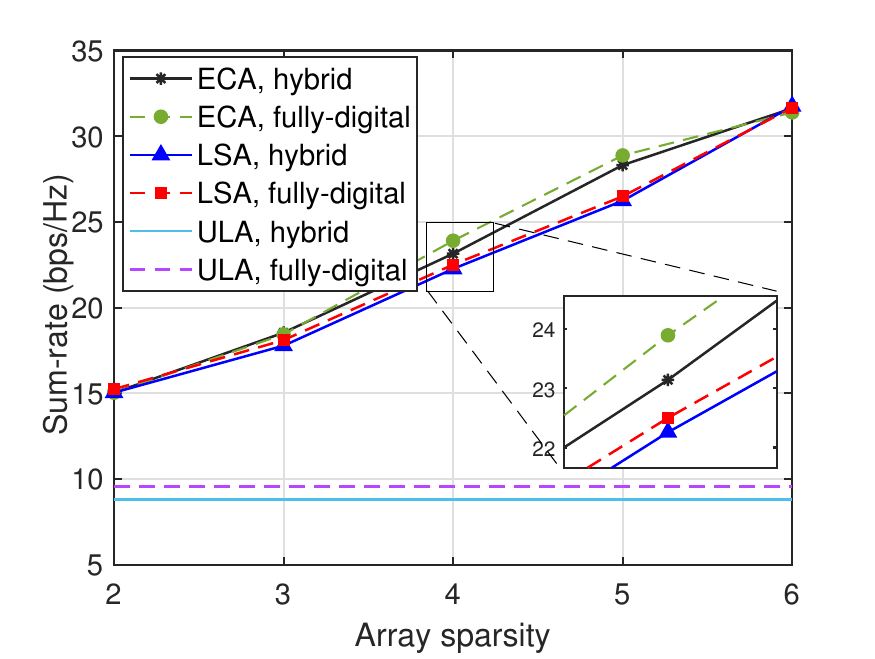}
		\caption{Scenario $2$: achievable sum-rate versus array sparsity.}
		\label{fig_10}
		\vspace{-0.6cm}
	\end{center}	
\end{figure}

\vspace{-0.3cm}
\subsection{Scenario 2: Users at the Same Angle}
Then, we consider the second scenario where the users are uniformly distributed at the same angle $\theta = 0$, with range in $ r \in [15~{\rm m}, 30~{\rm m}] $. Fig. \ref{fig_10} plots the achievable sum-rate versus the SA sparsity given the (approximately) same number of antennas. Again, the LSA and ECA systems still achieve much higher rates than the ULA systems thanks to the greatly enlarged array aperture and hence the near-field benefits. Next, the achievable rates of both SAs increases with the array sparsity, since a higher array sparsity leads to a smaller beam-depth which is particularly helpful for reducing the IUI when the users at the same angle. In addition, for both SAs, there is a small gap between their digital and hybrid beamforming, due to the residual IUI under the proposed analog beamforming. Other observations are similar to those for Fig.~\ref{fig_9} and thus are omitted due to limited space.

\vspace{-0.3cm} 
\subsection{Scenario 3: Users at Grating-lobes}
Last, we consider a more challenging scenario where the users are distributed at the grating-lobe directions ($\sin\theta \in \{-2/3, 2/3\}$) of both the LSA and ECA with their rang in $ r\in [20~{\rm m}, 30~{\rm m}]$. The array configuration is as follows: the LSA has $ Q_{\rm LSA} = 125 $ antennas with $ U = 3 $ and the ECA has $ Q_{\rm ECA }= 125 $ antennas with  $ \{M=5,N=3,L=18\} $. 
In Fig. \ref{fig_11}, we show the achievable sum-rate versus the transmit SNR for different arrays. The main observations are summarized as follows. First, for the LSA, the achievable rate by the proposed hybrid beamforming design suffers a considerable rate loss as compared to its optimal digital beamforming. This is intuitively expected, since the grating-lobes of LSA have high beam-heights (see Lemma~\ref{LemmaLSAGratingn-lobe}) and thus incur high IUI when occurred. In contrast, given the same antenna number, the ECA can greatly suppress the beam-heights of grating-lobes (see Proposition~\ref{ProBeam-height-ECA}), hence leading to a higher rate. Next, the rate gap between the ECA and LSA are increasing with the transmit SNR. This indicates that for this challenging scenario, the ECA generally yields better rate performance over the LSA in the high-SNR regime. 
\begin{figure}[!t]
	\begin{center}		
		\includegraphics[width=2.7in]{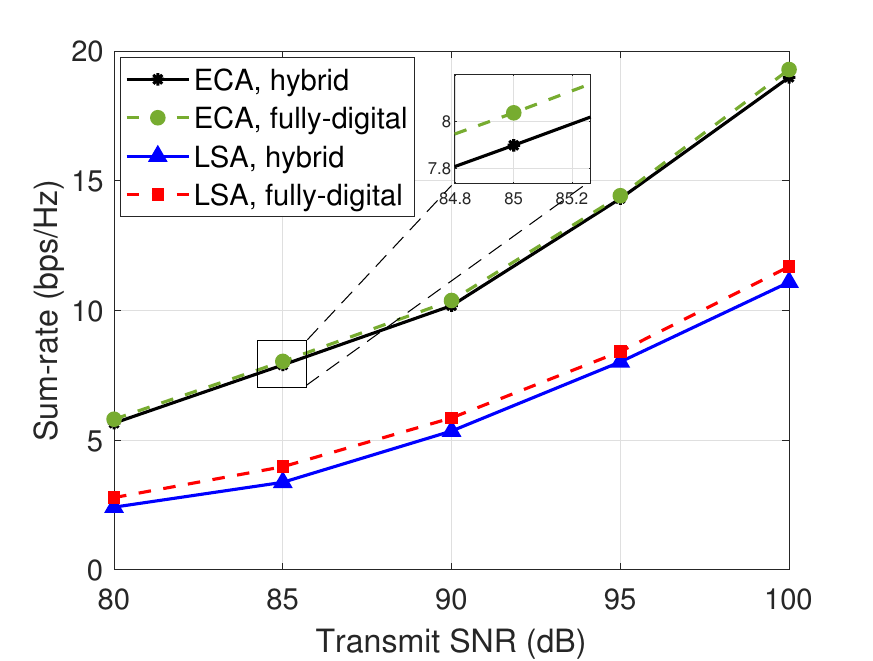}
		\caption{Scenario $3$: achievable sum-rate versus transmit SNR.}	
		\label{fig_11}
		\vspace{-0.4cm} 
	\end{center}
\end{figure}
\section{Conclusions}
In this paper, we proposed  two types of SAs to enable near-field communications for reducing the hardware and energy cost of existing ULA-based XL-array communication systems. First, we considered the LSA and characterized its near-field beam pattern. We showed that the LSA can enhance the beam-focusing gain, but it also introduces potential strong IUI due to the  undesired grating-lobes of comparable beam power with the main-lobe. Then an efficient hybrid beamforming design was proposed for the LSA to deal with the IUI issue. Next, we proposed the LSA for enabling near-field communications. It was revealed that the ECA can greatly suppress the beam power of near-field grating-lobes as compared to the LSA with the (approximately) same array sparsity, although it has a larger number of grating-lobes. A customized two-phase hybrid beamforming design was then proposed for the ECA. Finally, numerical results were presented to demonstrate the rate performance gain of the proposed two SAs over the conventional ULA.
\section*{Appendix A: Proof of Lemma \ref{LSAgeneral}}
Let $ \eta_{1} \triangleq  U\Delta$ and $ \eta_{2} \triangleq \frac{( U d_0)^{2}}{\lambda}( {\frac{{{{\cos }^2}{\theta _0}}}{{{r_0}}}} - \frac{{{{\cos }^2}\theta }}{{r}}) $. By  assuming $ \eta_{2}\neq 0 $, i.e.,$ \frac{\cos^2 \theta}{r} \neq  \frac{\cos^2 \theta_0}{r_0} $ , we have
\begin{equation*}
	{\small
	  \begin{aligned}
		f_{\rm LSA}({r_0},{\theta _0};r,\theta )&\overset{(b_1)}{\approx}\frac{1}{Q_{\rm LSA}} \sum_{q\in \mathcal{Q}_{\rm LSA}} e^{\jmath \pi\left(\eta_{1} q + \eta_{2} q^2 \right)} \\
		 &\overset{(b_2)}{\approx}\frac{1}{Q_{\rm LSA}} \int_{-Q_{{\rm LSA}} / 2}^{Q_{{\rm LSA}} / 2} e^{\jmath \pi\left(\eta_{1} q + \eta_{2} q^2\right)} {\rm d} q \\
		 &~=\frac{e^{-\jmath\pi\frac{\eta_{1}^{2}}{4\eta_{2}}}}{Q_{\rm LSA}} \int_{-{Q_{\rm LSA}}/ 2}^{{Q_{\rm LSA}} / 2} e^{\jmath \pi \eta_{2}\left(q+\frac{\eta_{1}}{2 \eta_{2}}\right)^2} {\rm d} q ,
	\end{aligned}}
\end{equation*}
where $(b_1)$ is due to Fresnel approximation,  $(b_2)$ approximately holds when $ \eta_{1} $ is not an even integer, i.e., $ \Delta\neq \frac{2u}{ U}, u=\pm 1, \dots, \pm  (U-1)$. 

If $ \eta_{2} > 0 $, we have
\vspace{-5pt}
\begin{equation*}
	{\small
	\begin{aligned}
		f_{\rm LSA}({r_0},{\theta _0};r,\theta ){\approx}  \frac{e^{-\jmath\pi\frac{\eta_{1}^{2}}{4\eta_{2}}}}{Q_{\rm LSA}} \int_{\sqrt{2\eta_{2}}(-\frac{ Q_{\rm LSA}}{2}+\frac{\eta_{1}}{2\eta_{2}})}^{\sqrt{2\eta_{2}}(\frac{Q_{\rm LSA}}{2}+\frac{\eta_{1}}{2\eta_{2}})} e^{\jmath \frac{1}{2}\pi t^2} {\rm d} t ,
	\end{aligned}}
\end{equation*}
by setting $ \frac{1}{2}t^{2} = \eta_{2}(q + \frac{\eta_{1}}{2\eta_{2}})^{2}$. 
Then, by defining $ \beta_{1} = \frac{\eta_{1}}{\sqrt{2\eta_{2}}} = \frac{\Delta}{{\sqrt{d_0 \Phi  }}} $ and $ \beta_{2} = \sqrt{2\eta_{2}}\frac{Q_{\rm LSA}}{2} = \frac{Q_{\rm LSA} U}{2}\sqrt{{d_0}\Phi  }$, we have
\vspace{-5pt}
\begin{equation*}
	f_{\rm LSA}({r_0},{\theta _0};r,\theta )\approx {e^{-\jmath\pi\frac{\Delta^{2}}{2d_0\Phi}}} [\widehat{C}(\beta_{1},\beta_{2}) +  \jmath(\widehat{S}(\beta_{1},\beta_{2})]/(2\beta_2),
		\vspace{-3pt}
\end{equation*}
where $\widehat{C}(\beta_{1},\beta_{2})$ and $ \widehat{S}(\beta_{1},\beta_{2}) $ are defined in Lemma~\ref{LSAgeneral}.   
Otherwise, if $ \eta_{2} < 0 $, by re-defining $ \beta_{1} = \frac{\eta_{1}}{\sqrt{-2\eta_{2}}} = \frac{\Delta}{{\sqrt{d_0 (-\Phi)  }}} $ and $ \beta_{2} = \sqrt{-2\eta_{2}}\frac{Q_{\rm LSA}}{2} = \frac{Q_{\rm LSA} U}{2}\sqrt{{d_0}(-\Phi)  }$, we have 
\vspace{-5pt}
\begin{equation*}
		f_{\rm LSA}({r_0},{\theta _0};r,\theta )\approx {e^{-\jmath\pi\frac{\Delta^{2}}{2d\Phi}}} [\widehat{C}(\beta_{1},\beta_{2}) -  \jmath(\widehat{S}(\beta_{1},\beta_{2})]/(2\beta_2).
		\vspace{-2pt}
\end{equation*}
Combining the above two cases leads to the desired result.
\vspace{-0.4cm}
\section*{Appendix B: Proof of Proposition \ref{LSA GL Beam-depth}}
For the 3-dB beam-depth of LSA grating-lobes, we have 
{\small $\beta_{2}=\frac{Q_{\rm LSA} U}{2}\sqrt{{d_0}| \frac{{{{\cos }^2}\theta_{0} }}{{r_{0}}} - {\frac{{{{\cos }^2}{\theta_{{\rm GL},u}^{(\rm G)} }}}{{{r}}}}| } \le 1.6.$}
Then, we have	{\small$|{\frac{{{{\cos }^2}{\theta_{{\rm LSA},u}^{\rm (G)}}}}{{{r}}}} -\frac{\cos \theta_{0}}{r_0}| \le \frac{4\varphi_{3\rm dB}^{2}}{Q_{\rm LSA}^{2} U^{2}d_0} \triangleq \frac{1}{r_{\rm LSA}^{'}}, $}
where $ \varphi_{3\rm dB} = 1.6$.

Next, if $r \ge \frac{r_0\cos^2{\theta^{\rm (G)}_{{\rm LSA},u} }}{\cos^2 \theta_0}  \triangleq r_{\rm e} $, the range of $ r $ is given by $r_{\rm e} \le r \le r_{0}r_{{\rm LSA},u}^{(\rm G)}/(r_{\rm LSA}-r_{0}),$
where $r_{\rm LSA} \triangleq r_{\rm LSA}^{'}\cos^{2}\theta$ and $ r_{{\rm LSA},u}^{\rm (G)} \triangleq r_{\rm LSA}^{'}\cos^2\theta^{\rm(G)}_{{\rm LSA},u}$.
Otherwise, if $r \le r_{\rm e}  $, the range of $ r $ is given by $r_{0}r_{{\rm LSA},u}^{(\rm G)}/(r_{\rm LSA}+r_{0}) \le r \le r_{\rm e}.$

Therefore, if $r_{0} \le r_{\rm LSA}$ , we have $\frac{r_{0}r_{{\rm LSA},u}^{(\rm G)}}{r_{\rm LSA}+r_{0}} \le r \frac{r_{0}r_{{\rm LSA},u}^{(\rm G)}}{r_{\rm LSA}-r_{0}},$
i.e., $ {\rm BD}_{{\rm LSA}, u}^{\rm (G)} = \frac{r_{0}r_{{\rm LSA},u}^{(\rm G)}}{r_{\rm LSA}-r_{0}} - \frac{r_{0}r_{{\rm LSA},u}^{(\rm G)}}{r_{\rm LSA}+r_{0}} = \frac{2r_{0}^{2}r_{{\rm LSA},u}^{\rm (G)}}{r_{\rm LSA}^{2}-r_{0}^{2}}$.
If $r_{0} \ge r_{\rm LSA}  $, we have $ \frac{r_{0}r_{{\rm LSA},u}^{(\rm G)}}{r_{\rm LSA}+r_{0}} \le r \le \infty,$
i.e., $ {\rm BD}_{{\rm LSA}, u}^{\rm (G)} = \infty,$ thus completing the proof.
\bibliographystyle{IEEEtran}
\bibliography{IEEEabrv.bib}

\end{document}

%% file: header.tex
\newtheorem{definition}{\emph{\underline{Definition}}}

\newtheorem{lemma}{\emph{\underline{Lemma}}}

\newtheorem{proposition}{\emph{\underline{Proposition}}}

\newtheorem{remark}{\bf \emph{\underline{Remark}}}

\def\l{\left}
\def\r{\right}
\def\({\left(}
\def\){\right)}

\def\b0{{\mathbf{0}}}

\newcommand{\nn}{\nonumber}